\documentclass[12pt, a4paper,reqno]{amsart}
\usepackage[margin=1in]{geometry}
\usepackage{amssymb}
\usepackage{enumerate}
\usepackage{hyperref}
\usepackage{doi,url}
\usepackage{dsfont}
\usepackage{mathabx}
\usepackage[normalem]{ulem}

\usepackage[dvipsnames]{xcolor}

\definecolor{purple}{rgb}{.5,0,1}
\definecolor{orange}{rgb}{1,.5,0}
\definecolor{pink}{rgb}{1,0,.5}

\newcommand{\nn}{\notag}

\usepackage{fancyvrb}
\usepackage{amsmath,amsthm}
\usepackage{enumitem}
\usepackage{thmtools}
\usepackage[style=apalike,style=numeric,maxbibnames=99,firstinits=true,citestyle=numeric,url=true,doi=true,
backend=bibtex
]
{biblatex}
\renewbibmacro{in:}{}
\AtEveryBibitem{\clearfield{number}}
\DeclareFieldFormat*{volume}{\mkbibbold{#1}}
\DeclareFieldFormat[article]{pages}{#1}

\numberwithin{equation}{section}
\newtheorem{theorem}{Theorem}[section]

\newtheorem{lemma}[theorem]{Lemma}
\newtheorem{corollary}[theorem]{Corollary}

\newtheorem{definition}[theorem]{Definition}
\newtheorem{remark}[theorem]{Remark}

\newtheorem*{theorem*}{Theorem}
\newtheorem*{lemma*}{Lemma}
\newtheorem*{remark*}{Remark}

\DeclareMathOperator{\supp}{supp}
\DeclareMathOperator{\tr}{tr}
\DeclareMathOperator{\Ran}{Ran}

\DeclareMathOperator{\dist}{dist}

\DeclareMathOperator{\Rea}{Re}

\newcommand\R{\mathbb R}
\newcommand\N{\mathbb N}
\newcommand\C{\mathbb C}
\newcommand\Z{\mathbb Z}

\newcommand\cW{\mathcal{W}}

\newcommand\e{\mathrm{e}}

\renewcommand\P{\mathbb P}
\newcommand\E{\mathbb E}

\newcommand\cL{\mathcal{L}}

\newcommand\cN{\mathcal{N}}

\newcommand\cX{\mathcal{X}}
\newcommand\cY{\mathcal{Y}}
\newcommand\cH{\mathcal{H}}

\renewcommand{\d}{\mathrm{d}}

\newcommand{\pr}{\prime}

\newcommand\what{\widehat}
\newcommand\wtilde{\widetilde}

\newcommand\beq{\begin{equation}}
\newcommand\eeq{\end{equation}}
\newcommand\be{\begin{equation}\begin{aligned}}
\newcommand\ee{\end{aligned}\end{equation}}

\newcommand{\abs}[1]{\left\lvert #1 \right\rvert}
\newcommand{\norm}[1]{\left\lVert #1 \right\rVert}
\newcommand{\scal}[1]{\left\langle #1 \right\rangle}
\newcommand{\set}[1]{\left\{ #1 \right\}}
\newcommand{\pa}[1]{\left( #1 \right)}

\newcommand{\hnorm}[1]{\left\{ \!\left\{ #1\right\}\! \right\}}

\newcommand{\cl}[1]{\lceil #1 \rceil}

\newcommand\La{\Lambda}

 \newcommand{\eq}[1]{\eqref{#1}}

\newcommand{\up}[1]{^{\left(#1\right)}}

\newcommand{\qtx}[1]{\quad\text{#1}\quad}
\newcommand{\mqtx}[1]{\; \ \text{#1}\; \  }
\newcommand{\sqtx}[1]{\;\text{#1}\;}

\newcommand{\bD}{\boldsymbol{\Delta}}

\newcommand{\tfd}{\pa{1- \tfrac{1}{\Delta}}}
\newcommand{\fd}{1- \frac{1}{\Delta}}
\newcommand{\sfd}{1- \tfrac{1}{\Delta}}
\newcommand{\nfd}{\pa{1- \frac{1}{\Delta}}}

\newcommand{\prr}{{\pr\pr}}

\begin{document}

\title[Slow propagation of information on  the random XXZ  spin chain]{Slow propagation of information on the  random XXZ  quantum spin chain}

\author{Alexander Elgart}
\address[A. Elgart]{Department of Mathematics, Virginia Tech; Blacksburg, VA, 24061, USA}
 \email{aelgart@vt.edu}

\author{Abel Klein}
\address[A. Klein]{Department of Mathematics, University of California, Irvine;   
Irvine, CA 92697,  USA}
 \email{aklein@uci.edu}

\thanks{A.E. was  supported in part by the NSF under grant DMS-1907435 and the Simons Fellowship in Mathematics Grant 522404.}


\begin{abstract} 
The random XXZ  quantum spin chain   manifests localization  (in the  form  of quasi-locality) in any fixed energy interval, as previously proved by the authors.  In this article it is shown that  this property implies slow propagation of information, one of the putative signatures of   many-body localization (MBL), in the same energy interval.

\end{abstract}
\keywords{Many-body localization, MBL, random XXZ spin chain, quasilocality}

\subjclass[2000]{82B44, 82C44, 81Q10, 47B80, 60H25}

\maketitle
\tableofcontents	

\section{Introduction}

The folk wisdom in physics is that particle interactions tend to delocalize (or, more precisely, dynamically thermalize) an isolated quantum system. In contrast, the presence of disorder in  the single particle context leads to the emergence of localization. It is, therefore, an interesting question of how such systems behave in the presence of both disorder and interactions. It has been proposed in the physics literature that in dimension one strong disorder leads to the so-called many-body localized (MBL) phase,  presumed to be  characterized by several exotic properties, such as the absence of thermalization, slow propagation  of information,  zero-velocity Lieb-Robinson bound,  Poisson distribution for level statistics, and area-law entanglement of eigenstates. This led to significant theoretical and experimental work in condensed matter physics over the last decade that focused on this phenomenon and its implications (see the physics reviews \cite{NandHuse,AL,abanin2019}). 

Let us stress that, as of today,  there is no unifying physics theory for MBL as well as no clear consensus among the physics community on the existence and stability of an MBL phase in the thermodynamic limit, even in the strong disorder regime, due to new numerical evidence and some theoretical work \cite{SBPV,SBPV1,Kiefer,SelsP,ABANIN2021,sierant20,morningstar,sierant22}.  Moreover, the hierarchical relationship between the proposed properties  (i.e., whether one of them implies others) is also not clear. The most significant difficulty in analyzing such models is that the dimension of the underlying Hilbert space grows exponentially fast with the system size $L$. Such growth limits  reliable numerics to small $L$ and makes it exceedingly hard to capture the rare but potentially critical events (such as resonances) that inevitably occur as the system size increases.

One of the central questions among the physics community is  what  a  suitable (and malleable) definition of MBL should be.  A popular choice there is the existence of Local Integrals of Motion (LIOM), which in particular implies a form of dynamical localization \cite{nachtergaeleslow}. However, the existence of LIOM is based on the exact diagonalization of the entire   Hamiltonian, a very strong assumption (it implies, in particular, the absence of a phase transition for the infinite system - a debatable assertion even among physicists).

In \cite{EK22}, we  introduced and  proved  a suitably defined notion of quasi-locality associated with  the finite XXZ spin-$\frac12$ random chain in any fixed energy interval  in a certain parameter region, that includes the limiting cases of strong disorder and weak interactions. The disordered XXZ model is one of the most common  models used  in the physics and mathematics literature for the study of  MBL  (e.g., \cite{abanin2019}).
We  consider finite volume Hamiltonians, which are what is typically discussed in the physics literature   {(e.g., \cite{abanin2019}).}    An  important feature of our result is that while the parameter region depends on the  energy interval, it is independent of system size.  

 A fixed energy interval is  sometimes referred to in physics as the zero temperature regime and has to be contrasted with the infinite-temperature regime, that is,  the whole energy spectrum\footnote{Note that  the diameter  of  the spectrum of the interacting system grows with the system size.}. While we argued in \cite{EK22} that our quasi-locality  property (different from any  of the physics signatures of MBL mentioned above) is very natural from a mathematical point of view,  we do not expect it to be useful  for studying the infinite-temperature MBL.

While the quasi-locality property is  amenable to rigorous analysis,  we did not attempt in \cite{EK22} to explore its connection to the putative manifestations of MBL proposed in the physics literature and mentioned above. The current work shows that the quasi-locality property implies slow propagation of information (one of the aforementioned signature properties) in the same energy interval on which the quasi-locality holds. This implication was not obvious to us when \cite{EK22} was completed, as it addresses a  different object and its proof required a new set of ideas.

Let us mention that the other proposed indicators of MBL-type localization (besides slow information propagation) seem to be  either unaccessible or significantly harder to reach:

\begin{enumerate}
\item It is expected that generic quantum many-body systems exhibit thermalization or even satisfy the eigenstate thermalization hypothesis. However, up to now mathematically tangible arguments have not been found for proving thermalization or its failure outside the realm of exactly solvable systems.  
\item As we already mentioned, the LIOMs' approach hinges on a complete (for all energies) localization of the underlying Hamiltonian. Such property has been only achieved for some exactly solvable models. Similarly, the proposal that the zero velocity Lieb-Robinson bound holds for random systems (see \cite{gogolin})  relies on   complete localization. A modified analogue of the Lieb-Robinson bound  was proven in  the droplet spectrum ( the special interval at the bottom of the spectrum) of the random XXZ spin chain  \cite[Theorem 3]{EKS3}. However, the form of the resulting bound  is interval-dependent, making it  unlikely  to be a  suitable candidate for an MBL characteristic for systems where a phase transition could potentially occur.

\item
In \cite{BeW1}, it was shown that, even without disorder, the excited states of the XXZ spin chain (with  pretty much any choice for the background potential) satisfy the area law with logarithmic corrections for any fixed energy interval. While the area law without these corrections holds in the  droplet spectrum of the random XXZ spin chain  \cite{BeW1}, it is not expected to persist beyond this interval on physical grounds \cite{BN,gogolin}. Thus it is hard to  probe localization  at higher energies using an area law-type criterion alone.

\item Another proposal is to link  localization of the interacting
ground state with exponential decay of the zero temperature  grand-canonical truncated correlations of local operators \cite{Mas2}. This result involves multiple limits (including zero temperature and thermodynamic ones), so it is not clear how to formulate it as a statement that holds for a finite system. In addition,  it appears that the result is highly sensitive to the  order in which these limits are taken. 

\end{enumerate}

Let us give an informal account of our result
(the formal  statement can be found in Section \ref{secmodel} below). The random XXZ quantum spin-$\frac 12$ chain on the finite discrete   interval (i.e.,   an interval in $\Z$) $\La_L=[1,L]$ is  given by the Hamiltonian
$H^L_\omega=H^L_0+\lambda V^L_\omega$
acting on   $\bigotimes_{i\in \La_L} \C_i^2$  (here $\C_i^2$ is a copy of $\C^2$), where
\[
H_0=\sum_{i=1}^{L-1}{\tfrac{1}{4}\pa{{I}-\sigma_i^z\sigma_{i+1}^z}-\tfrac{1}{4\Delta}\pa{\sigma_i^x\sigma_{i+1}^x+\sigma_i^y\sigma_{i+1}^y}}
\] (here $\sigma^{x,y,z}$ are  the standard Pauli matrices and $\Delta>1$ is the anisotropy   parameter), and 
 $V^L_\omega=\sum_{i=1}^L \omega_i \mathcal{N}_i$ is the random field (here $\mathcal{N}= \tfrac{1}{2} (I-\sigma^z)$ ,  $\omega = \set{\omega_i}_{i\in\Z}$  is a family of independent identically distributed positive  random variables with  sufficiently regular randomness, and 
$\lambda >0$ is the disorder parameter).

The following theorem is an informal statement of our main result, Theorem \ref{thm:localmodell}.

\begin{theorem*}[ Slow propagation of information, informal]\label{thm:localmodellinf}
 For a given energy $E>0$, there exists a non-trivial region in the $(\Delta,\lambda)$ parameter space, such that for any fixed point in this region, scales $L,\ell\in\N$, and all $t\in\R$,  the following holds: 
For every  observable   $\mathcal O$   supported on a discrete interval  $[a,b] \subset \La_L$, 
there exists an observable ${\mathcal O}_t={\mathcal O}(t,E,\ell  ,L)$, supported on   $[a-  c_E\ell,b +  c_E\ell]\cap\La_L $,   such that
\be\label{eq:sltr}
\E\norm{P_{[0,E]}\pa{\e^{itH^L_\omega}\mathcal O\e^{-itH^L_\omega}-{\mathcal O}_t}P_{[0,E]}}\le   C_{E}\norm{\mathcal O}(|t|+1)^{q_E} L^{{\xi_{E} }}\e^{- \theta_{E}\, \ell},
\ee
where $\E$ stands for the expectation with respect to $\omega$,   $P_{[0,E]}$ is the spectral projection of $H^L_\omega$ onto $[0,E]$, and $c_E, \xi_E,\theta_E >0$.
 \end{theorem*}

In general, the  Lieb-Robinson bound   for local spin Hamiltonians \cite{LR,NachtSims}  implies that there is an effective light cone for two-point dynamical correlations for such systems, meaning that these correlations propagate no faster than linearly in time, up to  exponentially small corrections. As a consequence of this bound,  if  $\mathcal O$ is a local observable supported on the discrete interval  $[a,b]$, given $\ell \in \N$ there exists an  observable $\mathcal O_{t,\ell}$, supported on the discrete interval $[a-\ell, b+ \ell]$, that approximates 
  the  (full)  Heisenberg evolution $\e^{itH}\mathcal O\e^{-itH}$: 
\[\norm{\e^{itH}\mathcal O\e^{-itH}-\mathcal O_{t,\ell}}\le C\norm{\mathcal O}e^{-m(\ell-v|t|)},\]
where  $m>0$  and $v>0$ is the velocity in the Lieb-Robinson bound
(see \cite{hastings2010locality}).

For  translation invariant systems, it is expected that  information can indeed spread within the light cone. This should be contrasted with the theorem above, that indicates a much slower rate at which the information spreads  for the random system: For a given value of   $\ell\gg\ln L$, it takes time $t\sim e^{c\ell}$ rather than $t\sim \ell$ until information can potentially escape the corresponding cone.  We note that our result considers propagation of observables within the energy window {  $[0,E]$}, and, in particular, is fully compatible with a possible phase transition for higher energies.

Rigorous results of this kind  have been previously obtained for exactly solvable systems (in fact, with no propagation at all), see, e.g., \cite{ARNSS}, and for the XXZ ferromagnetic spin chain studied here, but restricted to the energy interval $I_{\le 1}$  (introduced in Section \ref{secmodel}) corresponding to so called droplet spectrum   \cite[Theorem~2]{EKS3}.  Albeit  this result provides  a strong bound on the information propagation  speed, it is tailored for $I_{\le 1}$, and the method developed there cannot be adapted to the larger intervals $I_{\le q}$ that are handled in Theorem \ref{thm:localmodell}. While the dependence of our bound on both the system size and time is by no means optimal, it is generally expected that the random XXZ spin chain should exhibit some form of  slow propagation (e.g.,  \cite{abanin2019}).

 We now want to address the presence of the polynomial pre-factor in the volume size in \eqref{eq:sltr}. It is not unusual to have  a volume dependence in  local results concerning random systems  (e.g.,   the multiscale analysis  for random Schr\"odinger operators yields decay in a box of size $L$  for distances $ \ge L^\zeta$, $\zeta \in (0,1)$), as in the localization phenomenon there are two competing effects: A natural tendency for eigenstates to localize versus the small denominator problem, coming from resonances. The former is responsible for the exponential decay in  \eqref{eq:sltr},  whereas  the appearance of the volume prefactor is a manifestation of the latter.  Indeed, the number of resonances is directly related to the density of states, which, for the ferromagnetic XXZ spin system studied here, grows as a power of the volume (with the power increasing  with the energy). For this reason, we do not expect that the bound in  \eqref{eq:sltr} can be significantly improved.

Since the physics literature   mainly considers {\it finite}  spin  systems  (e.g.,  \cite{abanin2019}), the polynomial pre-factor is not a real issue in  the presence of the   exponential decay in $\ell$, as long as   $\ell \ge C\ln L$. 
 Let us also mention that in  (reliable)   numerical experiments on which physicists base their conclusions about the system's behavior, the typical value of $L$ does not exceed a few dozen, in which case our results fit in rather well with the physics picture for {\it all} energies.  

Nonetheless, the presence of the volume factor has the unfortunate side effect that it is not clear whether any conclusions can be drawn about the infinite volume  XXZ model from this estimate. This should be contrasted with the  random  Schr\"odinger operator  case, where the volume  dependence in  the decay estimates  can be overcome to yield Anderson localization in the infinite volume. 
 This disparity can be traced to the radical difference in the rank of the perturbation needed to decouple  the Hamiltonian into two (spatially non-interacting) parts in a discrete Schr\"odinger operator and in a spin chain. In the former case, the rank is comparable with the size of the boundary between these two parts (rank $2$ in one dimension), while in the latter case the rank is comparable with the dimension of the full Hilbert space, due to its tensor product nature.

One way to mitigate the influence of resonances is  to consider a  matrix element analogue of \eqref{eq:sltr}, that is,  to consider  $\abs{\langle \psi, P_{[0,E]}\pa{\e^{itH^L_\omega}\mathcal  O\e^{-itH^L_\omega}-{\mathcal O}_t}P_{[0,E]}\,\phi\rangle}$ for a pair of states $\psi,\phi\in\mathcal H$ instead of  $\norm{P_{[0,E]}\pa{\e^{itH^L_\omega}\mathcal  O\e^{-itH^L_\omega}-{\mathcal O}_t}P_{[0,E]}}$. As a consequence of  \eqref{eq:sltr}, we  obtain a bound on the expectation of this object where the dependence on the volume is  reduced to powers of $\ln L$ instead of $L$, see Corollary \ref{cor:localmodell} below. The price one has to pay here is that a priori the operator ${\mathcal O}_t$ constructed this way depends also on the states $\psi,\phi$.  We expect that this result  may  shed light on  properties of the infinite volume system.

This article is organized as follows: 
Section \ref{secmodel}  starts with the introduction of the XXZ quantum spin-$\frac 12$  chain in  a random field, followed by a short summary (Theorem~\ref{thmloc}) of our localization results in \cite{EK22},  which serves as  the starting point for   the statement of  our main result,  Theorem \ref{thm:localmodell},  exhibiting slow propagation of information under localization.  The statement for  matrix elements is given in Corollary \ref{cor:localmodell}. In Section \ref{sec:ir'} we introduce important  ingredients for the proof of Theorem \ref{thm:localmodell}.   Section~\ref{secprop}   contains the proof of Theorem \ref{thm:localmodell}.  Corollary \ref{cor:localmodell} is proven in Section \ref{secpropmatrix}.

Throughout the paper, we will use generic constants $C, c$, etc., whose values will be allowed to change from line to line, even  in a displayed equation. These constants will not depend on subsets of $\Z$, but    they will,  in general depend on  parameters of the model  such as $\mu$,  $\Delta_0$,  and $ \lambda_0$. When necessary, we will indicate the dependence of a constant on other parameters, say  $q$,  explicitly by writing the constant  as $C_q$, etc. These constants can always be estimated from the arguments, but we will not track the changes to avoid complicating the arguments.

\section{ The model, localization, and the  main result}\label{secmodel}

\subsection{Model description}

Let $\uparrow\rangle: = \begin{pmatrix} 1  \\ 0  \end{pmatrix} $ and  $\downarrow\rangle = \begin{pmatrix} 0  \\ 1 \end{pmatrix} $ denote  the elements of the canonical basis of $\C^2$, called   spin-up  and spin-down, respectively.  Let $\sigma^{x,y,z}$ be  the standard Pauli matrices, $\sigma^\pm=\frac 12( \sigma^x \pm i \sigma^y)$.  Set  $\mathcal{N} = \tfrac{1}{2} ({I}-\sigma^z)$, an operator on $\C^2$, and note that   $\mathcal{N} \uparrow\rangle=0$ and   $\mathcal{N} \downarrow\rangle=\downarrow\rangle$.  We  interpret $\downarrow\rangle$ as a particle, so $\cN$  is the projection onto the spin-down state (or local number operator). 

Let $\cH_i=\cH_{\set{i}}= \C^2_i$ for $i\in \Z$.  Given a vector $v\in \C^2$, we denote by $v_i$ its copy in $\cH_i$.
  If $T$ is an observable (i.e., operator) on $\C^2$, we denote by $T_i$ the observable $T$ acting on $\cH_i$.

The (infinite volume) XXZ quantum spin-$\frac 12$  chain in  a random field is informally given by the Hamiltonian 
\beq \label{infXXZ}
H_\omega=H_0+\lambda V_\omega,
\eeq
acting on   $\bigotimes_{i\in \Z} \cH_i$, where:
\begin{enumerate}
\item The  (disorder) free  Hamiltonian $H_0$ is given by 
\begin{align} \label{hii}
H_0&=\sum_{i\in\Z} \pa{\tfrac{1}{4}\pa{{I}-\sigma_i^z\sigma_{i+1}^z}-\tfrac{1}{4\Delta}\pa{\sigma_i^x\sigma_{i+1}^x+\sigma_i^y\sigma_{i+1}^y}}\\ \nn &
=\sum_{i\in\Z} \pa{\tfrac{1}{4}\pa{{I}-\sigma_i^z\sigma_{i+1}^z}-\tfrac{1}{2\Delta}\pa{\sigma_i^+\sigma_{i+1}^-+\sigma_i^-\sigma_{i+1}^+}},
\end{align}
where    $\Delta>1$ is the anisotropy   parameter, specifying the Ising phase ($\Delta=1$ selects the Heisenberg chain and  $\Delta=\infty$ corresponds to the  the Ising chain).

\item    $V_\omega=\sum_{i\in\Z} \omega_i \mathcal{N}_i$ is the random field, where
 $\omega = \set{\omega_i}_{i\in\Z}$  is a family of independent identically distributed random variables, whose  common probability  distribution $\mu$  is  absolutely continuous with a bounded density  and satisfies
\beq\label{mu}
  \set{0,1}\subset \supp \mu\subset[0,1] ,
  \eeq 
and $\lambda >0$ is the disorder parameter.

\end{enumerate}

We set $\omega_S= \set{\omega_i}_{i\in S}$ for $S\subset \Z$, and denote the corresponding expectation and probability by $\E_S$ and $\P_S$. 
     (We will mostly  omit the subscript  and just write $\E$ and $ \P$  when the choice of $S$ is clear from the context.)

The  (infinite volume) Hamiltonian $H_\omega$ in \eq{infXXZ} can be rigorously defined on an appropriately defined Hilbert space, 
but in this work we   only  consider finite volume Hamiltonians,  since that is what is typically discussed in the physics literature.

 Given  a finite subset $\Lambda$ of $\Z$ ($\La$ will always denote a finite subset),  
 we consider the finite dimensional  Hilbert space  $\mathcal H_\Lambda=\otimes_{i\in \Lambda} \cH_i$.   If $A\subset \La$ and $T$ is an operator on $\cH_A$, we consider $T$ as an operator on $\mathcal H_\Lambda$ by identifying it with the operator $T\otimes I_{\cH_{\La \setminus A}}$ acting
on $\cH_\La=\cH_A \otimes\cH_{\La \setminus A}$. 
 (For a fixed $\La$ we will often omit $\La$ from the notation, e.g.,     $A^c=\La \setminus A$.) For $S\subset \Z$ we let $\abs{S}$ denote the cardinality of the set $S$.

Since
\beq
 \tfrac 14\pa{{I}-\sigma_i^z\sigma_{i+1}^z}=\tfrac 1 2\pa{\cN_i+\cN_{i+1}}-\cN_i\cN_{i+1},
\eeq
 we set
\beq\label{tildeh}
{h}_{i,i+1}= - \mathcal{N}_{i}\mathcal{N}_{i+1} -\tfrac{1}{2\Delta}\pa{\sigma_i^+\sigma_{i+1}^-+\sigma_i^-\sigma_{i+1}^+},
\eeq
a  self-adjoint operator ${h}_{i,i+1}$  on the four-dimensional Hilbert space  $\cH_{\set{i,i+1}}=\cH_i \otimes \cH_{i+1}$.
An explicit calculation shows  
 \beq\label{nth}
\norm{{h}_{i,i+1}}=1.
\eeq

We  can rewrite $H_0$ as
\beq\label{tildeH0}
 H_0= \sum_{i\in\Z} \pa{{h}_{i,i+1}+ \tfrac 1 2 (\mathcal{N}_{i} + \mathcal{N}_{i+1}) }=\sum_{i\in\Z} {h}_{i,i+1}+ \cN^\Z, \qtx{where} \cN^\Z=\sum_{i\in\Z} \cN_i,
 \eeq
which leads naturally to our definition of finite volume Hamiltonians.

\begin{definition}
The random XXZ quantum spin-$\frac 12$  chain  on a finite subset $\Lambda$ of $\Z$ is given by the self-adjoint  Hamiltonian
\be
H^\Lambda= H_0^\Lambda + \lambda V_\omega^\La \qtx{acting on} \cH_\La,
\ee
where
\be
H_0^\Lambda&=\sum_{\set{i,i+1}\subset \La} {h}_{i,i+1}+ \cN^\La,   \qtx{with} \cN^\La=\sum_{i\in\La} \cN_i \qtx{and}
 V_\omega^\La = \sum_{i\in\La} \omega_i \mathcal{N}_i.
\ee
We set  $R_z^\La= (H^\La -z)^{-1}$   for $z \not\in \sigma(H^\La)$, the resolvent of   $H^\La $.
\end{definition}

The  free Hamiltonian $H_0^\Lambda$ can be rewritten as
\beq\label{eq:modH}
H_0^\Lambda=-\tfrac{1}{2\Delta} \bD^\Lambda + \cW^\Lambda \qtx{on} \cH_\La,
\eeq
where 
\begin{align}\label{bD}
 \bD^\Lambda = \sum_{\set{i,i+1}\subset \Lambda} \pa{\sigma_i^+\sigma_{i+1}^-+\sigma_i^-\sigma_{i+1}^+} \qtx{and}
 \cW^\Lambda = \cN^\Lambda -\sum_{\set{i,i+1}\subset \Lambda} \cN_i\cN_{i+1} .
 \end{align}

 The canonical (orthonormal) basis $\Phi_\La $  for $\cH_\La$ is constructed as follows: 
 Let $\phi_\emptyset= \Omega_\Lambda=\otimes_{i \in \Lambda}\uparrow\rangle_i$  be the vacuum state. Then
 \be \label{eq:standbasis}
& \Phi_\La 
=\set{\phi_A=\pa{\prod_{i\in A}\sigma_i^-}\Omega_\Lambda:\ A\subset \Lambda} =\bigcup_{N=0}^{\abs{\La}}  \Phi_\La\up{N}, 
 \ee
 where $\Phi_\La\up{N}= \set{\phi_A:  A\subset \Lambda, \ \abs{A}=N}$.  We remark that $ \Phi_\Lambda\up{0}=\set{\Omega_\Lambda}$.

The total (spin-down) number operator $\cN^\La$ on $\La$ is diagonalized by the canonical basis:  $\cN^\La \phi_A= \abs{A}\phi_A $ for $A\subset \La$, and hence has 
 eigenvalues $0,1,2,\ldots, \abs{\La}$.  We set  $\cH_\Lambda\up{N}=\Ran\pa{\chi_N(\mathcal N^\Lambda)}$,  obtaining   the Hilbert space decomposition 
$ \cH_\La= \bigoplus_{N=0}^{\abs{\La}} \cH_\La\up{N}$. 

 The operators $\cW^\La$ and $V^\Lambda_\omega$  are  also diagonalized by the canonical basis, and hence the operators $\cN^\La$, $\cW^\La$, and $V^\Lambda_\omega$  commute.  $\cW^\La$ is the number of clusters operator:     $\cW^\La \phi_A=W_{A}\phi_A $ for $A\subset \La$, where $W_{A}$ is the number of connected components (clusters) of $A$ as a subset of $\La$, so 
 $\sigma\pa {\cW^\La}\subset \set { 0,1,2,\ldots, \abs{ \La}}$.  $V_\omega^\La$ is the random field: 
 $V_\omega^\La\phi_A=\omega_A \phi_A $ for $A\subset \La$, where $\omega_A=\sum_{i\in A}\omega_i$.

 The Hamiltonian   $H^\Lambda$  preserves the total particle number,  
\beq
[H^\Lambda,\cN^\Lambda]= -\tfrac{1}{2\Delta}[ \bD^\Lambda,\cN^\Lambda]  =0, 
\eeq
a  feature  that makes the XXZ model especially amenable to analysis.

It can be verified (e.g., \cite{EK22}), that 
\beq\label{H0W}
\tfd \cW^\Lambda\le H_{0}^\Lambda , { \qtx{so}} \tfd \cW^\Lambda\le H^\Lambda,
\eeq
and  the spectrum of $H^\Lambda$ is  of the form  
\beq
\sigma(H^\Lambda)=\set{0} \cup \pa{\left[1 -\tfrac 1 \Delta, \infty \right ) \cap  \sigma(H^\Lambda) }.
\eeq
Moreover, 
the lower bound in \eq{H0W}  suggests the introduction of the energy thresholds $k\tfd$, $ k=0, 1,2\ldots$.

\subsection{Localization as quasi-locality}
 Henceforth,  by a subset of $\Z$ we will always mean a finite subset and by  an interval in $\Z$  a connected nonempty subset of  $\Z$.    The observable $\mathcal O$ is  said to have  support in $A\subset \La$ (we write $\supp {\mathcal O}=A$) if $\mathcal O$ acts trivially on $\cH_{A^c}$, that is, $\mathcal O=\mathcal O_A\otimes {I}_{\mathcal H_{A^c}}$ where $\mathcal O_A$ is an observable on $A$. (We will  identify $\mathcal O$ with $\mathcal O_A$.)   Note that the support of an operator is not uniquely defined.

  Given $\emptyset \ne S\subset    \Z$, we  define the orthogonal projections  $P_\pm^{S}$ on  $\cH_S$ by 
 \begin{align}\label{P+-S}
P_+^{S}=\bigotimes_{i\in S}\pa{{I}_{\mathcal H_{i}}-\cN_i}= \chi_{\set{0}}\pa{\cN^{S}}  \qtx{and} 
P_-^{S}= {I}_{\mathcal H_{S}}-P_+^{S}=  \chi_{[1,\infty)}\pa{\cN^{S}}.
\end{align} 
$P_+^{S}$  is the orthogonal projection onto states with no particles in the set $S$;  $P_-^{S}$  is the orthogonal projection onto states with at least one particle in $S$. (Note that $P_-^{\set{i}}=\cN_i$ for $i\in  \Z$.)
We also set 
 \be \label{P+-empty}
 P_+^{\emptyset } ={I}_{\cH_S} \qtx{and } P_-^{\emptyset }=0.
 \ee

 Given  $J\subset \R$ measurable, $B(J) $ denotes the collection of Borel measurable functions that vanish outside $J$;  we set $B_1(J)= \set{f\in B(J) :  \sup \abs{f} \le 1}$.

In \cite{EK22} we interpreted localization for the  random XXZ quantum spin-$\frac 12$  chain  as a form of quasi-locality.  The following theorem follows immediately from  \cite[Theorem 2.4 and Corollary 2.6]{EK22}.

\begin{theorem}[Quasi-locality]\label{thmloc} 
 Fix  $\Delta_0>1$ and  $ \lambda_0 >0$.   Then for all   $R\ge 0$  there exist  constants $D_R,F_R, \wtilde \xi_R, \wtilde\theta_R>0$ (depending on $R$, $\Delta_0$, $ \lambda_0$) such that,
 for all $\Delta \ge \Delta_0$ and $\lambda \ge \lambda_0$ with  $\lambda \Delta^2\ge D_R$,  finite interval  $\Lambda\subset  \Z$, and  $A\subset B\subset \La$ with    $A$ connected in $\La$,  we have the following:
 
 \begin{enumerate}
\item   For all  $z\in \C$ with  $ \Rea z \le R\tfd$  we have
\be\label{eq:mainbnd5}
\E_\La\set{\norm{P_-^{A}R_z^\Lambda P_+^{B}}^{\frac 14}}\le F_R\abs{\Lambda}^{\wtilde\xi_R} \e^{-\wtilde\theta_R\dist_\La(A,B^c)}.
\ee

\item 
\be\label{eq:eigencor5} 
\E_{\Lambda}\pa{\sup_{{f\in B_1\pa{\left ( -\infty, R\nfd \right ]}}}\norm{P_-^{A}f(H^\Lambda) P_+^{B}}}\le  F_R\abs{\Lambda}^{\wtilde\xi_R} \e^{-\wtilde\theta_R\dist_\La(A,B^c)}.\ee  
\end{enumerate}
\end{theorem}

\begin{remark}
\cite[Theorem 2.4]{EK22}  is stated and proved  for $R=k+\tfrac 3 4$, where $k\in \N^0$,
 and  real energies  $E\le (k+\tfrac 3 4)\nfd$.  However,  the proof of \cite[Theorem 2.4]{EK22} is also valid for complex energies $z$ with  $ \Rea z  \le (k+\tfrac 3 4)\nfd$, with the same constants. Picking $k\in \N^0$ so that $R\le k+\tfrac 3 4$ yields the result stated above. (As an alternative,  the proof of \cite[Theorem 2.4]{EK22}    can be adapted for the case
$R=k + \beta $ with $\beta\in (0,1)$;
we fixed $\beta=\frac 34$ in \cite{EK22}   for simplicity.) 
\end{remark}

\begin{remark}\label{remconn} If $A$ is not connected in $\La$, the theorem still holds with \eq{eq:mainbnd5} replaced by
\be\label{eq:mainbndnot5}
\\E_\La\set{\norm{P_-^{A}R_z^\Lambda P_+^{B}}^{\frac 14}}\le  F_R  {\Upsilon}^\La_A \abs{\Lambda}^{\wtilde\xi_R} \e^{-\wtilde\theta_R\dist_\La(A,B^c)},
\ee 
where ${\Upsilon}^\La_A$ denotes the number of connected components of $A$ in $\La$. This follows from  \eq{eq:mainbnd5}  and 
\beq\label{Pconnectedcomp}
P_-^{A}=\sum_{j=1}^{{\Upsilon}^\La_A} P_+^{\bigcup_{i=i}^{j-1} A_i}P_-^{A_j},
\eeq
where
  $A_j$, $j=1,2,\ldots, {\Upsilon}^\La_A$,  are  the  connected components of $A$ in $\La$.

\end{remark}

\subsection{Slow propagation of information}

Our main result shows that localization  of the XXZ random spin chain in a fixed energy interval implies 
slow propagation of information in this interval.

In this article we will assume that $\Delta_0>9$ in Theorem~\ref{thmloc}. This is done   to simplify our analysis,  but  in fact the result holds for arbitrary $\Delta_0> 1$ with minor modifications of the proofs.\footnote{For $1<\Delta_0 \le  9$ we need to improve the decay rate $m_0$ in \eq{CT}, which is derived from the lower bound in \eq{eq:hatH1'}. If $1<\delta_0<\Delta_0 \le  9$, we would have to  replace $\what H_k^{ \Lambda}$ in the proof by  $\what H_{k+r}^{ \Lambda}$, where $r=  \cl{\frac 1 {\delta_0-1}-\frac 1 8}$,  leading to $m_0= \ln \pa{(r + \frac 1 8) \pa{\delta_0-1}}>0 $.}

Given  $ q \in \tfrac 12 \Z$, we  consider the  energy intervals  
 \be\label{Ikle2}
I_{\le q}&= \left(-\infty, (q+\tfrac 3 4)\tfd\right] ,\quad 
I_q= \left[\sfd, (q+\tfrac 3 4)\tfd\right],\\
\check I_{\le q}&=\left(-\infty, (q+\tfrac 7 8)\tfd\right] ,\quad 
\check I_q= \left[\sfd, (q+\tfrac 78)\tfd\right].
\ee 
These intervals are increasing with $q$.    We also note the relation    
\be\label{indq}
  \check I_{\le q-1} \subset  I_{\le q-\frac 1 2}   \subset \check I_{\le q-\frac 1 2}\qtx{for}  q \in \tfrac 12 \N .
\ee

 Let  $\La$  be a finite subset of $\Z$. Given an interval $I\subset \R$, we set   $P_{ I}=P^\La_{ I}=\chi_{ I}(H^\La) $. If $T$ and $Y$ are  observables on $\cH_\La$,  we define
\beq
\pa{T}_{Y}= YT  Y^*.
\eeq
We also consider the  Heisenberg time evolution of observables:
\be
\tau_t^\La(T)= \e^{itH^\La}T\e^{-itH^\La}\qtx{for} t\in \R.
\ee

  Given  $ M\subset \Lambda\subset \Z$,   we let
   \be
{[M]^\La_s } &:=\begin{cases}\set{x\in\La: \dist_\La\pa{x,M}\le s} &\sqtx{if}  s\in \N^0=\set{0} \cup \N  \\ \set{x\in\La:\dist_\La\pa{x,  M^c}\ge 1-s}= M\setminus [M^c]^\La_{-s}&\sqtx{if}  s\in-\N
\end{cases}.
\ee
If $M=\set{j}$ we write ${[j]^\La_q }={[\set{j}]^\La_q }$.

\begin{definition} Given $q\in \frac 12 \N^0$,  we say that Condition $\cL_q$ is satisfied  if
the  parameters $\Delta$ and $\lambda$ satisfy the hypotheses of Theorem~\ref{thmloc}  for $R= q+\frac 78$,
so the conclusions of the theorem are valid for $R= q+\frac 78$, that is,  on the energy interval $\check I_{\le q}$.  
We set $\xi_q=\wtilde \xi_{ q+\frac 78}$ and $\theta_q=\wtilde \theta_{q+\frac 78}$.
 \end{definition}

We use the notation $\langle t\rangle= \pa{1 + t^2}^{\frac 12}$ for $t\in \R$.  For $ q\in  \tfrac 12\N^0$ we set $\hat q=\cl{q}$ and  define $\beta_q$ recursively by 
\be\label{betaq}  
\beta_0=0,\quad \beta_q=\beta_{q-\frac 12} +9\hat q + 13.
\ee 
Note that  $ 9 q^2 + \frac {61}2 q  \le \beta_q\le 9 q^2 +  \frac {79}2 q $.

\begin{theorem}[ Slow propagation of information]\label{thm:localmodell}
 Let   $q\in \frac 1 2 \N$,  and assume  Condition $\cL_{q+\frac 12}$ is satisfied.  Then there exists a constant $C_q$  such that for any given finite interval $\La \subset \Z$,  scale $\ell\in \N$, and all $t\in\R$,  the following holds:
For  every  observable  $T$   supported on an interval $\mathcal X\subset \La$ with $\norm{T}\le 1$
there exists an observable ${T}_t={T}(t,q,\ell,\La)$, supported in  $ [\cX]^\La_{(13+\beta_{q+\frac 12} )\ell}$,   such that
\be\label{eq:locality2mgq}
\E\norm{\pa{\tau^\La_t(T)-{T}_t}_{P^\La_{I_{\le q}}}}\le   C_{q}\langle t\rangle^{2q+4} \abs{\Lambda}^{{\xi_{q+\frac 12} }}\e^{- \theta_{q+\frac 12}\, \ell}.
\ee
 \end{theorem}
 
 This theorem  is proved  in Section~\ref{secprop}.

As  discussed in the introduction, we  can improve the dependence on the volume by considering matrix elements instead of the norm.  The constant  $c_\mu>0$ in the estimate comes from   the large deviation estimate \cite[Eq. (3.50]{EK22}) and depends only on the probability distribution $\mu$.
 
 \begin{corollary}[Slow propagation of information, matrix elements version]\label{cor:localmodell}
 Let   $q\in \frac 1 2 \N$,  and assume  Condition $\cL_{q+1}$ is satisfied. 
There exist constants $C_q$ and $Y_q$  such that given a finite interval $\La \subset \Z$,  scale $\ell\in \N$,  $t\in\R$,  the following holds:  Given subsets $M_i\subset \La$, $i=1,2$, with $\abs{M_1}=\abs{M_2}$, then 
for  every  observable  $T$   supported on an interval $\mathcal X\subset \La$, with $\norm{T}\le 1$
and $\abs{\cX} \le   \ln \La$,
there exists an observable  ${T}_t={T}(t,q,\ell,\La,  M_1\cup M_2)$, supported on   $ [\cX]^\La_{(13+\beta_{q+1} )\ell}$,   such that
\be\label{eq:locality2mgqcor}
\E\norm{\pi_{M_1}\pa{\tau^\La_t(T)-{T}_t}_{P^\La_{I_{\le q}}}\pi_{M_2}} \le C_{q}\langle t\rangle^{2q+5} \pa{\ln \abs{\La}}^{\xi_{q+1} }\e^{-\frac 12 \min\set{ \theta_q,\theta_{q+\frac 1 2}, c_\mu}\ell},
\ee
provided $ \abs{\La} \ge Y_q$.
 \end{corollary}
 
 For fixed   $M_i\subset \La$, $i=1,2$, with $\abs{M_1}=\abs{M_2}$, the bound \eq{eq:locality2mgqcor}  improves on the dependence on $\abs{\La}$ in  \eq{eq:locality2mgq}, but the observable ${T}_t$ in \eq{eq:locality2mgqcor}  a-priori depends on  $M_1 \cup M_2$.
 Note also that if $\abs{M_1}\ne\abs{M_2}$ the left hand side of \eq{eq:locality2mgqcor}  equals $0$.

  The proof of this corollary is given  in Section~\ref {secpropmatrix}.

\section{Key  proof ingredients}\label{sec:ir'}

In this section we collect a number of definitions, statements and   lemmas that,  in conjunction  with Theorem \ref{thmloc}, will facilitate the proof of Theorem \ref{thm:localmodell}.  

\subsection{Preliminaries}
$\La$ will always denote a finite subset of $\Z$ and $A\subset \La$ will always denote an interval.

 Given  $ M\subset \Lambda$ and  $s\in \N$,  we set
\be
\partial_{s}^{\La, out} M&:=\set{x\in\Lambda:\ \dist_\La\pa{x,M}=s}= [M]^\La_s\setminus M,\\
 \partial_{s}^{\La,in} M&:=\set{x\in\Lambda:\ \dist_\La \pa{x, M^c}=s}=M\setminus [M]^\La_{-s},\\
 \partial_s^\La M &:=  \partial_{s}^{\La, in} M \cup  \partial_{ex}^{\La, out}M=[M]^\La_s \setminus  [M]^\La_{-s},\\
 \pmb{\partial}^\La M &:=  \set{\set{x,y}\subset \La: \  (x,y) \in \pa{\partial_{1}^{\La,in} M \times \partial_{1}^{\La,out}M } \cup \pa{\partial_{1}^{\La,out} M \times \partial_{1}^{\La,in}M }}.
 \ee
If $s=1$, we occasionally omit it from the notation altogether.

 Given $ B\subset \N^0$, we set $Q_B^\Lambda=\chi_{B}\pa{\cW^\Lambda}$,   $Q_m^\Lambda= Q_{\set{m}}^\Lambda$ for $m \in \N^0$, 
   and note that $Q_0^\Lambda=  P_+^\Lambda$ and  $Q_\N^\Lambda= \chi_{\N}(\cN^\Lambda)$.  For $k\in \N$, we set 
\be\label{QkhatQ}
Q_{\le k}^\Lambda   =Q_{\set{1,2,\ldots,k}}^\Lambda =\sum_{ m=1}^k Q_m^\Lambda \qtx{and} 
\what Q_{\le k}^\Lambda   =Q_{\le k}^\Lambda + \tfrac {k+1} k Q_0^\Lambda.
\ee
We also set   $Q^\La_{>k }= I - \what  Q^\Lambda_{\le k}=\chi_{[k+1, \Lambda]}\pa{\cW^\La}$.  For 
 $k\in \N$ we have (see  \cite[Lemma~3.5]{EK22})
\begin{align}\label{trXk}
\tr {Q_{\le k}^{\Lambda}}&\le {k} \abs{\Lambda}^{2k}  \qtx{and}   \tr \chi_{\check  I_{\le k}}(H^\La)\le  k\abs{\Lambda}^{2k}+1.
\end{align}

We also set  
\be \label{eq:compH'}
\what  H_0^{ \Lambda}&=H^{\Lambda}+\tfd Q_0^\Lambda,\\
\what  H_k^{ \Lambda}&=H^{\Lambda}+{k}\tfd \what  Q_{\le k}^{\Lambda} \qtx{for}  k\in \N.
\ee 
We use the notation
\be
\what  R^{\Lambda}_{k,z}&= \pa{\what  H_k^\La  -z}^{-1}  \mqtx{for} z\notin \sigma(\what  H_k^\La) \qtx{for}   k\in \N^0,
\ee 
and recall the resolvent identity
 \be\label{eq:resmodl}
R_{z}^{\Lambda}=\what  R^{\Lambda}_{k,z}+ k\tfd R^{\Lambda}_{z} \what Q^\Lambda_{\le k}\what  R^{\Lambda}_{k,z}=\what  R^{\Lambda}_{k,z}+ k\tfd\what  R^{\Lambda}_{k,z}\what Q^\Lambda_{\le k}  R^{\Lambda}_{z}.
\ee

 It follows from \eqref{H0W} and \eq{Ikle2} that   for $k\in \N^0$ we have
 \be \label{eq:hatH1'}
\what  H_k^{ \Lambda}\ge  \pa{k+1} \tfd {I}   \qtx{and}  \pa{\what  H_k^{ \Lambda}-E}  \ge \tfrac 1{8}\tfd  {I}\mqtx{for} E \in \check I_{\le k}
\ee
and
\be
\norm{\what  R^{\Lambda}_{k,z}} \le \norm{\what  R^{\Lambda}_{k,\Rea z}} \le 8 {\tfd}^{-1} \qtx{for}  \Rea z \in \check I_{\le k}.
\ee

For $q\in \frac 12 \N^0$, we set 
\be
\what  H_q^{ \Lambda}= \what  H_{\what q}^{ \Lambda} \qtx{and } \what  R^{\Lambda}_{q,z}=\what  R^{\Lambda}_{\what q,z}.
\ee

\subsection{Quasi-locality (deterministic)}

\begin{lemma}[{\cite[Lemma~3.1]{EK22}}] \label{lem:locdet'} 
Let $T$ be an operator on the Hilbert space $\cH_{\Lambda}$,  and let $Y$ be  a projection  on  $\cH_{\Lambda}$ such  that  $[Y,T]=0$ and $[Y, P_\pm^K]=0 $ for all $K\subset \La$.

Suppose
\begin{enumerate}

\item  For all   $K \subset \La $   we have   $[P_-^{K},T]P_+^{[K]_1^{\La}}=0$.

 \item For all  connected  $K \subset\La $  we have  $\norm{[P_-^{K},T]}\le  \gamma  $. 
 
 \item  $T_Y$, the restriction of the operator $T$ to $\Ran Y$, is invertible with
$ \norm{T_Y^{-1}}_{\Ran Y} \le \eta^{-1}$,  where $\eta>0$.

 \end{enumerate}

Then for all $A\subset B\subset \La$, 
we have
 \be \label{eq:localitybndet}
\norm{P_-^{A}\,T_Y^{-1}\,P_+^{B}}_{\Ran Y}\le  \gamma^{-1} \e^{-m  \dist_\La  (A,B^c)}= \eta^{-1}  \e^{-m ( \dist_\La  (A,B^c)-1)}, \sqtx{with} 
m=\ln\pa{{\gamma^{-1}}\eta }.
\ee 
\end{lemma}

This lemma yields quasi-locality for the resolvent of the operators $H^\La$ and $\what  H_q^{ \Lambda}$, as discussed in  \cite[Section~3.2]{EK22}.  The operator $ T= \what  H_q^{ \Lambda}-z$ satisfies the hypotheses of Lemma~\ref{lem:locdet'}  for  $q\in \frac 12 \N$ and $  \Rea z \in \check I_{\le \what q}$, with 
  $\gamma = \frac 1 {\Delta}\le  \frac 1 {\Delta_0}$,  $Y={I}_{\cH_\La}$,  and  
  $\eta= \dist (z,\sigma(\what  H_q^{ \Lambda}))\ge \frac 1{8}\tfd\ge \frac 1{8}\pa{1-\frac 1 {\Delta_0}}$, and hence   for 
 $A\subset B\subset \La$    the estimate  \eq{eq:localitybndet} yields  (recall  we assumed $\Delta_0 > 9$)
 \be  \label{CT}
\norm{P_-^{A}\what  R^{\Lambda}_{q,z}P_+^{B}}\le \tfrac 1{\Delta_0}  \e^{-m_0 \dist_\La  (A,B^c)},
 \mqtx{where} m_0= \ln  \tfrac {\Delta_0-1} 8>0.
\ee

\subsection{Consequences of quasi-locality}

\begin{lemma}\label{lem:detbnP}
Fix  $k\in\N$.   Given   a collection   $\set{S_i}_{i=1}^{k+1}$ of  nonempty subsets   of  $\Lambda$ with
\beq\label{dSiSj}
\min_{i\neq j }\dist_\La \pa{S_i,S_j}\ge 2\ell +1, \qtx{where}\ell \in \N,
\eeq 
 we have
\be\label{PQPLD}
\E \Big\|{P_{ I_{\le  k}}\prod_{i=1}^{k+1} P_-^{S_i}}\Big\|\le  C_k \Upsilon^\La_{max} \abs{\La}^{2k+1} \e^{-m_0\ell},
\ee
 where $\Upsilon^\La_{max} =\max_{i=1}^{k+1} {\Upsilon}^\La_{{S_i}}$ (see Remark~\ref{remconn}).
\end{lemma}

\begin{proof} 
Note that  $P_{ I_{\le  k}}\prod_{i=1}^{k+1} P_-^{S_i}=P_{ I_{k}}\prod_{i=1}^{k+1} P_-^{S_i}$.
We can represent $P_{ I_{k}}$ as a contour integral
\be
P_{ I_{ k}}=\frac1{2\pi i}P_{ I_{ k}}\oint_\Gamma R_{z}^{\Lambda}dz,
\ee
where $\Gamma$ is defined by $\Gamma=\set{z\in\C:\ \min_{x\in  I_{ k}}|x-z|=\frac18\tfd}$. Note that  $\norm{P_{ I_{k}}R_{z}^{\Lambda}}\le \frac  8{\fd}$ for any $z\in\Gamma$.

Using  \eqref{eq:resmodl}   and \eqref{eq:hatH1'}, we deduce that
\be
P_{ I_{ k}}=\frac{k\tfd }{2\pi i}P_{ I_{ k}}\oint_\Gamma R^{\Lambda}_{z} \what Q^\Lambda_{\le k}\what  R^{\Lambda}_{k,z}dz.
\ee

Let   $\Theta_{k,\ell}=  \chi_{[0,2\ell +k]}\pa{\cN^\La}$.  Note that  $[H^\La, \Theta_{k,\ell}]=0$
and   $[P_\pm^B, \Theta_{k,\ell}]=0$ for $B\subset \La$.   Moreover, it follows from \eq{dSiSj} that
\beq\label{ThetaPPP}
\Theta_{k,\ell} \prod_{i=1}^{k+1} P_-^{[S_i]^\La_{\ell}} =\Theta_{k,\ell} Q^\La_{>k }\prod_{i=1}^{k+1}  P_-^{[S_i]^\La_{\ell}}.
\eeq

Since  \eqref{CT} and \eq{Pconnectedcomp} yield
 \be
\Big\|P_+^{[S_i]^\La_{\ell}}\what  R^{\Lambda}_{k,z}P_-^{S_i}\Big\|\le C \Upsilon^\La_{{S_i}} \e^{-m_0\ell}\qtx{for} i=1,2,\ldots k+1,
\ee
we have, using $P_+^{[S_i]^\La_{\ell}}+ P_-^{[S_i]^\La_{\ell}}= {I}_{\cH_\La}$, that
\be\label{PQPi2}
\Big\|{ \what Q^\Lambda_{\le k}\what  R^{\Lambda}_{k,z}\prod_{i=1}^{k+1} P_-^{S_i}\Theta_{k,\ell}}\Big\| 
&  \le C (k+1) \Upsilon^\La_{max}\e^{-m_0\ell}+\Big\|{ Q^\Lambda_{\le k}Q_{>k }\Theta_{k,\ell}\prod_{i=1}^{k+1}  P_-^{[S_i]^\La_{\ell}} \what  R^{\Lambda}_{k,z}}\Big\| \\
&=  { C} (k+1)\Upsilon^\La_{max} \e^{-m_0\ell}= C_k  \Upsilon^\La_{max}\e^{-m_0 \ell},
\ee 
where we used \eq{ThetaPPP}  and  $ Q^\Lambda_{\le k}Q_{>k }=0$.

To prove \eq{PQPLD}, 
by a large deviation estimate (see \cite[Eqs (5.18)--(5.23)]{EK22}) we have 
\be\label{PneT}
&\P\set{P_{I_{\le k}}\ne  P_{I_{\le k}}\Theta_{k\ell} } = \P\set{   P_{I_{\le k}} \chi_{(2\ell +k,\abs{\La}]}\pa{\cN^\La} \ne 0     }\\  & \quad =
\P\set{ \sigma\pa{H^\La \chi_{(2\ell +k,\abs{\La}]}\pa{\cN^\La}}\cap I_{\le k}\ne \emptyset} 
\le C_{k} \abs{\La}^{2k+1} \e^{-d_\mu \ell},
\ee
where $d_\mu$ depends only on the probability distribution $\mu$.  The estimate \eq{PQPLD} follows  by assuming, without loss of generality, that $m_0\le d_\mu$.
\end{proof}

 \begin{remark}   We have  the following consequence of Lemma~\ref{lem:detbnP} and \eqref{eq:eigencor5}.
Consider    $i, j_1,j_2,\ldots, j_{k}\in \La$ such that     $ \min_{s\ne r}   \abs{j_r-j_s}  \ge 2 \ell +1 $
and $\min_{s} \abs{i-j_s}\ge 3 \ell+1$. Then
\be
\E \pa{\sup_{\substack{f\in B(I_{ k}):\\ \|f\|_\infty\le1}} \norm{ \cN_{i}  f(H^\La)   \cN_{j_1}  \cN_{j_2} \ldots \cN_{j_{k}} } } &\le C_k \abs{\Lambda}^{\max\set{\xi_k,2k+1}} \e^{-\theta_k \ell}.
\ee
For $k=1$ this bound has been established in \cite{EKS1}.
It is proved as follows:
\be \notag
  \cN_{i}  f(H^\La)   \cN_{j_1}  \cN_{j_2} \ldots \cN_{j_{k}} = \cN_{i}  f(H^\La) P_+^{[i]_\ell}  \cN_{j_1}  \cN_{j_2} \ldots \cN_{j_{k}}+\cN_{i}  f(H^\La) P_-^{[i]_\ell}  \cN_{j_1}  \cN_{j_2} \ldots \cN_{j_{k}}.
 \ee
 The expectation of the  first term (with the sup inside) is estimated by \eqref{eq:eigencor5}.  The  second term  is estimated  by \eq{PQPLD} using   $f(H^\La)= f(H^\La)P_{ I_{ k}}$ for $f\in B(I_{ k})$.

\end{remark}

\begin{lemma}\label{lem:Mmod}
Let ${T}$ be an observable  supported on an interval $\mathcal X\subset\La$ with $\norm{T}\le 1$,  let $k,\ell\in\N$, and assume    $\dist \pa{\cX, \Z\setminus \La} > 9(k+1)\ell +1$.  Consider the observables
\be\label{cTj}
{T}_j={T} P_+^{\partial_{3\ell}[\mathcal X]^{\Lambda}_{9j\ell}} \prod_{i=1}^{j-1}P_-^{\partial_{3\ell}[\mathcal X]^{\Lambda}_{9i\ell}} \sqtx{for} j=1,2,\ldots, k+1, \sqtx{with}\prod_{i=1}^{0}P_-^{\partial_{3\ell}[\mathcal X]^{\Lambda}_{9i\ell}}={I}.
\ee 
Then    
\be\label{defTj}
{T}_j=P_+^{\partial_{3\ell}[\mathcal X]^{\Lambda}_{9j\ell}} {T}_j P_+^{\partial_{3\ell}[\mathcal X]^{\Lambda}_{9j\ell}}= P_+^{\partial^{out}_{3\ell}[\mathcal X]^{\Lambda}_{9j\ell}} \otimes   P_+^{\partial^{in}_{3\ell}[\mathcal X]^{\Lambda}_{9j\ell}}\otimes T\prod_{i=1}^{j-1}P_-^{\partial_{3\ell}[\mathcal X]^{\Lambda}_{9i\ell}},   
\ee
where $\supp T_j=[\mathcal X]^{\Lambda}_{(9j+3)\ell} \qtx{and} \supp\pa{ T\prod_{i=1}^{j-1}P_-^{\partial_{3\ell}[\mathcal X]^{\Lambda}_{9i\ell}}}= [\mathcal X]^{\Lambda}_{(9j-6)\ell}$,
and  
such that 
\be\label{T-Tj666}
\E \norm{ \pa{T-\sum_{j=1}^{ { k+1}}{T}_j}_{P_{I_{ \le k}}}}\le  C_k  \abs{\La}^{2k+1} \e^{-m_0\ell}.
\ee

\end{lemma}
\begin{proof}
We decompose
\be
{I}=\prod_{i=1}^{  k+1}P_-^{\partial_{3\ell}[\mathcal X]^{\Lambda}_{9i\ell}}\,+\,\sum_{j=1}^{  k+1}P_+^{\partial_{3\ell}[\mathcal X]^{\Lambda}_{9j\ell}} \prod_{i=1}^{j-1}P_-^{\partial_{3\ell}[\mathcal X]^{\Lambda}_{9i\ell}}.
\ee

It follows from  Lemma \ref{lem:detbnP},  using   $3\ell >2\ell +1$,  $\partial_{3\ell}[\mathcal X]^{\Lambda}_{9i\ell}\ne\emptyset$ for $i=1,2,\ldots,k+1$ since   $\dist \pa{\cX, \Z\setminus \La} > 9(k+1)\ell +1$, and  ${\Upsilon}^\La_ {{\partial_{3\ell}[\mathcal X]^{\Lambda}_{9i\ell}}} \le 2$,    that
 \be
\E \Big\|{\Big({T} \prod_{i=1}^{ k+1}P_-^{\partial_{3\ell}[\mathcal X]^{\Lambda}_{9i\ell}}\Big)_{P_{I_{ \le k}}} }\Big\|\le  C_k\abs{\La}^{2k+1} e^{- m_0\ell}.
\ee  

The estimate \eq{T-Tj666}  follows.
\end{proof}

\subsection{Decoupling}

 Let $A\subset \La$ be  an interval.  We consider the Hamiltonian
\be
 H^{A,A^c}=H^A + H^{A^c} \qtx{on} \cH_\La,
  \ee  
set $R_z^{A,A^c}=\pa{H^{A,A^c}-z}^{-1}$, and let
\be\label{eq:Gamma}
\Gamma^A= H^\La -  H^{A,A^c}= \sum_{\set{i,i+1}\in   \pmb{\partial}^\La A} h_{i,i+1}.
\ee 
It follows from \eq{tildeh}  
 that
\be\label{hPN}
\norm{P_+^{\set{i}}{h}_{i,i+1}}=\norm{P_+^{\set{i+1}}{h}_{i,i+1}}=  \tfrac  1 {2\Delta},
\ee
so
\beq\label{PGamma1}
 \norm{P_+^{A} \Gamma^A}\le \tfrac  1 \Delta  \qtx{and}   \norm{P_+^{A^c} \Gamma^A}\le \tfrac  1 \Delta.
 \eeq

We also set    $\tau_t^{A,A^c}(M)=\e^{it H^{A,A^c}}Me^{-it H^{A,A^c}}$, where $M$ is an observable on $\cH_\La$.

We fix an infinitely differentiable function  $\wtilde \Psi: \R \to [0,1]$ such that 
\be
\wtilde \Psi (u)=\begin{cases}     0 & \qtx{for} u \in (-\infty, -1] \\
1 & \qtx{for} u \in [ 0, {\tfrac 3 4} \tfd ]\\
0& \qtx{for} u \in [{ \tfrac 7 8} \tfd, \infty) 
\end{cases} ,
\ee
and set 
\be
\Psi_q(u)=\begin{cases}\wtilde  \Psi (u) & \qtx{for} u \in (-\infty, 0] \\
1 & \qtx{for} u \in [ 0, \pa{q + \tfrac 3 4} \tfd ] \\
\wtilde \Psi (u-q )& \qtx{for} u \in [\pa{q + \tfrac34} \tfd, \infty) 
\end{cases} \qtx{for} q \in \tfrac 12 \N.
\ee  
Note that  $\Psi_q$
is  an infinitely differentiable function  on the real line  such that  $0\le  \Psi_q \le 1$, and 
\be
\Psi_q(u)=\begin{cases}     0 & \qtx{for} u \in (-\infty, -1] \\
1 & \qtx{for} u \in [ 0, \pa{q + \tfrac 3 4} \tfd ] \\
0& \qtx{for} u \in [\pa{q + \tfrac 7 8} \tfd, \infty) 
\end{cases} .
\ee
We also define   $\Phi_{q,t}(u)= \Psi_q (u) \e^{itu}$ for $t\in \R$, so  $\Psi_q= \Phi_{q,0}$.  Note that
$\supp \Phi_{q,t} \subset [-1, \pa{q + \frac 7 8} \nfd]$.

We have
\be\label{MPPsi998}
P_{I\le q}= P_{I\le q}\Phi_q(H^\La)\le \Phi_q(H^\La) \le P_{\check I\le q}.
\ee
Thus, if $M$ is an   observable on $\cH_\La$,  we have
\beq
\norm{\pa{M}_{P_{I_{\le q}}}}   \le \norm{\pa{M}_{ \Psi_q} } \le   \norm{\pa{M}_{P_{\check I\le q}}}.
\eeq

\begin{lemma} \label{lemHS} Let   $q\in \frac 12\N$,  and assume Condition $\cL_q$. Then for  $t\in \R$,  $\ell \in \N$, $b\in \N$, and an interval $A\subset \La$,  we have
\be
\E \norm{\pa{\Phi_{q,t}(H^\La)-\Phi_{q,t}(H^{A,A^c})}P_+^{\partial_{b\ell} A}}\le C_q \scal{t}^3 \abs{\La}^{\xi_q}\e^{-b \theta_q \ell}.
\ee
\end{lemma}

\begin{proof}
We use the Helffer--Sj\"ostrand formula for {smooth}
functions $f$ of self-adjoint operators \cite{HeSj89,HuSi00}.  We consider  the norms
\begin{equation} \label{sdfn}
  \hnorm{f}_m := \sum_{r=0}^m \int_{\mathbb{R}}\!\mathrm{d}u\;
  |f^{(r)}(u)|\,(1 + \abs{u}^{2})^{\frac {r-1} 2}  , \quad  m=1,2,\ldots \,. 
\end{equation}
If $ \hnorm{f}_m < \infty$ with $m \ge 2$, then for any self-adjoint operator
$K$ we have
\begin{equation}\label{HS}
  f (K) = \int_{\R^{2}} \!\d\tilde{f}(z) \, (K-z)^{-1} ,
\end{equation}
where the integral converges absolutely in operator norm.  Here $z= x + i y$,
$\tilde{f}(z)$ is an \emph{almost analytic extension} of $f$ to the complex
plane, $\d\tilde{f}(z) := \frac 1 {2\pi}\partial_{\bar{z}}\tilde{f}(z)
\,\mathrm{d} x\, \mathrm{d} y $, with $\partial_{\bar{z}}= \partial_x + i
\partial_y$, and $|\d\tilde{f}(z)| := (2\pi)^{-1}
|\partial_{\,\overline{z}}\tilde{f}(z)| \,\mathrm{d} x\, \mathrm{d} y$.
Moreover, for all $p \ge 0$ we have
\begin{equation}\label{HShigherorder}
  \int_{\R^{2}} \! |\d\tilde{f}(z)| \;\frac{1}{|\mathrm{Im}\, z|^p}  \le c_p
  \  \hnorm{f}_m < \infty  \quad \text{for} \quad m \ge p+1
\end{equation}
with a constant $c_{p}$ (see \cite[Appendix B]{HuSi00} for details).

   By the  Helffer--Sj\"ostrand formula  we have (recall \eq{eq:Gamma})
\be
&\Phi_{q,t}(H^\La)-\Phi_{q,t}(H^{A,A^c})=  \int_{\R^{2}} \!\d\wtilde \Phi_{q,t}(z) \pa{ (H^\La-z)^{-1}-(H^{A,A^c}-z)^{-1}}
\\
& \quad =  - \int_{\R^{2}} \!\d\wtilde \Phi_{q,t}(z) {  (H^{A,A^c}-z)^{-1}\Gamma^A  (H^\La-z)^{-1}}.
\ee
Thus
\be
&\E \norm{\pa{\Phi_{q,t}(H^\La)-\Phi_{q,t}(H^{A,A^c})}P_+^{\partial_{b\ell} A}}\\
&  \quad \le \int_{\R^{2}}\abs{ \!\d\wtilde \Phi_{q,t}(z)} \E\norm{  (H^{A,A^c}-z)^{-1}\Gamma^A  (H^\La-z)^{-1}P_+^{\partial_{b\ell} A}}\\
&  \quad \le C \int_{[-1, \pa{q + \frac 7 8} \nfd] \times \R}\abs{ \!\d\wtilde \Phi_{q,t}(z)} \;|\mathrm{Im}\, z|^{-\frac74}\ \ \E\norm{   P_-^{\partial_{1} A}(H^\La-z)^{-1}P_+^{\partial_{b\ell} A}}^{\frac 14}\\
&  \quad \le C  \hnorm{\Phi_{q,t}}_3  \abs{\La}^{\xi_q}\e^{-b \theta_q \ell} \le C q \scal{t}^3 \abs{\La}^{\xi_q}\e^{-b\theta_q \ell}, 
\ee
where we used  $\Gamma^A =\Gamma^A  P_-^{\partial_{1} A}$ , \eq{HShigherorder},  
the fact that $\hnorm{\Phi_{q,t}}_3\le  C q \scal{t}^3 $ with a constant independent of $q$ by its construction, and Theorem~\ref{thmloc}.
\end{proof}

\begin{lemma}\label{lem:detachaA}  Let   $q\in \frac 12\N$,  and assume Condition $\cL_q$.  Let   $\ell \in \N$,  and consider  an interval
 $A\subset \La$. 
Let $T$ be an arbitrary observable with $\norm{T}\le1$. Then  for all $b\in \N$ we have
\be
&\E\norm{\pa{\tau_t \pa{ P_+^{\partial_{{b\ell}} A}TP_+^{\partial_{{b\ell}} A}}-  \Psi_q( H^{A,A^c})  \tau_t^{A,A^c}\pa{ P_+^{\partial_{{b\ell}} A}TP_+^{\partial_{{b\ell}} A}} \Psi_q( H^{A,A^c})    }_{ P_{I_{\le q}} }}\\
& \hskip40pt \le  C q \scal{t}^3 \abs{\La}^{\xi_q}\e^{-b \theta_q \ell}.
\ee
\end{lemma}
\begin{proof} 

We  have, recalling $P_{I_{\le q}}=\Phi_q(H^\La)P_{I_{\le q}}$, 
\be
&\norm{\pa{\tau_t \pa{ P_+^{\partial_{{b\ell}} A}TP_+^{\partial_{{b\ell}} A}} -    \Psi_q( H^{A,A^c})  \tau_t^{A,A^c}\pa{ P_+^{\partial_{{b\ell}} A}TP_+^{\partial_{{b\ell}} A}} \Psi_q( H^{A,A^c})    }_{ P_{I_{\le q}} }} \\
&  =  \left \|    \left(\Psi_q(H^\La)\tau_t \pa{ P_+^{\partial_{{b\ell}} A}TP_+^{\partial_{{b\ell}} A}}\Psi_q(H^\La) \right.  \right.\\
& \hskip110pt  -  \left.  \left. \Psi_q( H^{A,A^c})  \tau_t^{A,A^c}\pa{ P_+^{\partial_{{b\ell}} A}TP_+^{\partial_{{b\ell}} A}} \Psi_q( H^{A,A^c})   \right)_{ P_{I_{\le q}} }\right\| \\
&\le  \norm{{\Phi_{q,t}(H^\La)\pa{ P_+^{\partial_{{b\ell}} A}TP_+^{\partial_{{b\ell}} A}}\Phi_{q,-t}(H^\La) -    \Phi_{q,t}( H^{A,A^c})  \pa{ P_+^{\partial_{{b\ell}} A}TP_+^{\partial_{{b\ell}} A}}  \Phi_{q,-t}( H^{A,A^c}) }}.
\ee

Using Lemma~\ref{lemHS}, we conclude that
\be
&\E\norm{\pa{\tau_t \pa{ P_+^{\partial_{{b\ell}} A}TP_+^{\partial_{{b\ell}} A}}-  \Psi_q( H^{A,A^c})  \tau_t^{A,A^c}\pa{ P_+^{\partial_{{b\ell}} A}TP_+^{\partial_{{b\ell}} A}} \Psi_q( H^{A,A^c})    }_{ P_{I_{\le q}} }}\\
& \hskip40pt \le  C_q \scal{t}^3 \abs{\La}^{\xi_q}\e^{-b \theta_q \ell}.
\ee
\end{proof}

Note that 
 \be  \label{PsiPqM}
 &  \Psi_q( H^{A,A^c})= \Psi_q( H^{A,A^c})P_{{\check I_{\le q}}}(H^{A,A^c}), \qtx{so} \\
& \norm{\pa{\pa{M}_{ \Psi_q( H^{A,A^c})}}_{P_{{I_{\le q}}}}}\le  \norm{\pa{M}_{P_{{\check I_{\le q}}}(H^{A,A^c})} }.
\ee

\section{Proof of slow propagation of information}\label{secprop}

In this section we prove  Theorem~\ref{thm:localmodell} .  We start with the following lemma.

\begin{lemma}\label{thm:localmodell2}
 Let   $q\in \frac 12\N^0$,  and assume Condition $\cL_q$ is satisfied.  Then   there  exists a constant $C_q$  such that, for any given finite interval $\La \subset \Z$,  scale $\ell\in \N$, and all $t\in\R$, the following holds:

  \begin{enumerate}
\item  Let  $T$ be an observable  supported on an interval $\mathcal X\subset \La$ with $\norm{T}\le 1$, such that
\beq\label{P-T}
T= P_-^{\cX} T P_-^{\cX}.
\eeq
Then there exists an observable ${T}_t={T}(t,q,\ell)$, supported in $[\cX]^\La_{\beta_q \ell}$,   such that
\be\label{eq:locality2mg}
\E\norm{\pa{\tau^\La_t(T)-{T}_t}_{P^\La_{I_{\le q}}}}\le  C_q\langle t\rangle^{p_q}\,\abs{\Lambda}^{ \xi_{q}}\e^{-\theta_{q}\ell}.
\ee
where   $p_0=0$,  $p_q=2q+2$ for $q \in \frac 12\N$,  and $\beta_q$ is defined in  \eqref{betaq}.

\item  If $T$ is is an observable on $\La$ with  $\norm{T}\le 1$ of the form
\be\label{eq:Mspca}
T= P_+^{ [ \cY ]^{\La}_{2\ell -1}} P_-^{ \partial^{\La} [ \cY ]^{\La}_{2\ell}}\wtilde T   P_-^{ \partial^{\La} [ \cY ]^{\La}_{2\ell}} P_+^{ [ \cY ]^{\La}_{2\ell -1}}, \qtx{with}  \supp \wtilde T  ={ \partial^{\La} [ \cY ]^{\La}_{2\ell}},
\ee
where $\cY\subset \La$ is an interval, we can choose $T_t$, supported on $[\cY]^\La_{(\beta_q +2 )\ell+1}$ and satisfying \eq{eq:locality2mg},   such that
\be\label{eq:Mspca'}
{T}_t=P_+^\mathcal Y {T}_t P_+^\mathcal Y.
\ee 	
\end{enumerate}

 \end{lemma}

\begin{proof}  
The lemma is proved by induction on $q\in \frac 12\N^0$.  The lemma is obviously true for $q=0$ with $M_t=M$, $C_0=0$,     $p_0=0$,  $ \xi_0=0$, $\theta_0=0$.  Given $q\in \frac 12\N$, we assume the lemma is true for $q-\frac 1 2$ and  prove the lemma also holds for $q$.

So let  $\La \subset \Z$ be a finite interval. (We will often omit $\La$ from the notation.)
 We first  consider  an observable $M$ on $\La$  with $\norm{M}\le 1$ of the form
\be \label{TjM}
M=  P_+^{\partial^{out}_{3\ell} A} \otimes \wtilde M,
 \ee
 where $A\subset \La$ is an interval and  $ \wtilde M$ is an observable such that 
 \beq\label{TjM46}
 \supp  \wtilde M  = A ,\quad  \wtilde M =   P_+^{\partial^{in}_{3\ell} A}  \wtilde M P_+^{\partial^{in}_{3\ell} A},
 \eeq
 and
\be\label{TjM2}
\wtilde M =  P_-^{A} \wtilde M P_-^{A}.
\ee
Note that  $M= P_+^{\partial_{3\ell} A} M P_+^{\partial_{3\ell} A}$.

Lemma \ref{lem:detachaA} gives
\be\label{PPsiq}
&\E\norm{\pa{\tau^\La_t \pa{ M}-  \Psi_q( H^{A,A^c})  \tau_t^{A,A^c}\pa{ M   } \Psi_q( H^{A,A^c}) }_{P^\La_{I_{\le q}}} }\le  C q \scal{t}^3 \abs{\La}^{\xi_q}\e^{-3 \theta_q \ell}.
\ee

It follows from \eq{TjM}, \eq{TjM46} and \eq{TjM2} that
 \be
 \tau_t^{A,A^c}(M)= \tau_t^{A}(\wtilde M)\tau_t^{A^c}(P_+^{\partial^{out}_{3\ell} A})
 =P_-^A \tau_t^{A}(\wtilde M)  \tau_t^{A^c}(P_+^{\partial^{out}_{3\ell} A})P_-^A.
  \ee
  
Using
\beq
P_{\check I_{\le q}}^{A,A^c}P_-^{A}=  P_{\check I_{q}}^{A,A^c} P_{\check I_{q}}^{A} P_{\check I_{\le q-1}}^{A^c}P_-^{A}=P_{\check I_{\le q}}^{A,A^c} P_{\check I_{\le q}}^{A} P_{\check I_{\le q-1}}^{A^c}P_-^{A},
\eeq 
we get 
  \be\label{PAAcA}
&P_{\check I_{\le q}}^{A,A^c}\tau_t^{A}(\wtilde M)\tau_t^{A^c}(P_+^{\partial^{out}_{3\ell} A})P_{\check I_{\le q}}^{A,A^c}  = P_{\check I_{\le q}}^{A,A^c}P_-^A \tau_t^{A}(\wtilde M)\tau_t^{A^c}(P_+^{\partial^{out}_{3\ell} A})P_-^AP_{\check I_{\le q}}^{A,A^c} \\
& \qquad =P_{\check I_{\le q}}^{A,A^c}P_{\check I_{\le q}}^{A}\tau_t^{A}(\wtilde M)P_{\check I_{\le q}}^{A}\,P_{\check I_{\le q-1}}^{A^c}\tau_t^{A^c}(P_+^{\partial^{out}_{3\ell} A})P_{\check I_{\le q-1}}^{A^c}P_{\check I_{\le q}}^{A,A^c}.
 \ee

 We set
 \be\label{wtildeMt}
 {\wtilde M}_t = P_+^{\partial_{\ell}^{in}A}P_{\check I_{\le q}}^{A}\tau_t^{A}(\wtilde M)P_{\check I_{\le q}}^{A}P_+^{\partial_{\ell}^{in}A},
 \ee
 an observable with $\supp {\wtilde M}_t= A$.
Using  $\wtilde M =  P_+^{\partial^{in}_{3\ell} A}    \wtilde M P_+^{\partial^{in}_{3\ell} A}$,  $I- P_+^{\partial_{2\ell}^{in}A}=P_-^{\partial_{2\ell}^{in}A}$, and \eqref{eq:eigencor5}, we get
 \be\label{eq:incoq}
\E\norm{ P_{\check I_{\le q}}^{A}\tau_t^{A}(\wtilde M)P_{\check I_{\le q}}^{A}-P_{\check I_{\le q}}^{A}{\wtilde M}_tP_{\check I_{\le q}}^{A}}\le C_q \abs{\Lambda}^{ \xi_q}\e^{-2\theta_{q}\ell}. 
 \ee

To treat the term in $A^c$ in \eq{PAAcA}, note that $A^c=A^c_L \cup A^c_R$, where  
$A^c_L $ and $ A^c_R$ are subintervals of $\La$ (possibly empty, in each case they do not have to be considered);   $A^c_L $ is the interval to the left of $A$ and 
 $A^c_R $  the interval to the right of $A$. Since $\dist_\La \pa{A^c_L ,A^c_R}\ge  \abs{A}+1\ge 2$, we have  $H^{A^c}=H^{A^c_L}+H^{A^c_R}$, and   $P_+^{\partial^{out}_{3\ell} A}= P_+^{\partial^{out,L}_{3\ell} A}P_+^{\partial^{out,R}_{3\ell} A}$, where $\partial^{out,\#}_{3\ell} A =\partial^{out}_{3\ell} A\cap A^c_\#$ for $\#=L,R$. Thus we have
 \be\label{leftright}
 &P_{\check I_{\le q-1}}^{A^c}\tau_t^{A^c}(P_+^{\partial^{out}_{3\ell} A})P_{\check I_{\le q-1}}^{A^c}\\
&\quad  = 
 P_{\check I_{\le q-1}}^{A^c} \pa{P_{\check I_{\le q-1}}^{A^c_L}\tau_t^{A^c_L}(P_+^{\partial^{out,L}_{3\ell} A}) P_{\check I_{\le q-1}}^{A^c_L}}\pa{P_{\check I_{\le q-1}}^{A^c_R}\tau_t^{A^c_R}(P_+^{\partial^{out,R}_{3\ell} A}) P_{\check I_{\le q-1}}^{A^c_R}}P_{\check I_{\le q-1}}^{A^c}.
 \ee
 
The two expressions in parenthesis in \eq{leftright}  are treated the same way.    So let us consider the expression on the right.  Since  $\check I_{\le q-1}\subset  I_{\le q -\frac 12}$, we can  use the induction hypothesis on  the interval $A^c_R$. But it does not suffice to use it directly for the  observable $P_+^{\partial^{out,R}_{3\ell} A}$ because we seek an observable $Y_t=Y(t,q,\ell)$ on $A^c$  (we now drop the $R$ from the notation to simplify the exposition:  $A^c$ will  stand for $A^c_R$, $P_+^{\partial^{out,}_{3\ell} A}$ for $P_+^{\partial^{out,R}_{3\ell} A}$, etc.)  that not only satisfies an estimate like
\be\label{eq:obsO2a}
\E\norm{P_{\check I_{\le q-1}}^{A^c}\pa{\tau_t^{A^c}(P_+^{\partial^{out}_{3\ell} A})-Y_t}P_{\check I_{\le q-1}}^{A^c}}
\le C_{q-\frac 12}  \langle t\rangle^{p_{q-\frac 12}}\, \abs{A^c}^{ \xi_{q-\frac 12}}\e^{-\theta_{q-\frac 12}\ell},
\ee
but also satisfies
\be\label{eq:outco2}
Y_t=P_+^{\partial_{\ell}^{out}A} Y_t P_+^{\partial_{\ell}^{out}A}.
\ee

 Let $K(t)=\tau_t^{A^c}(P_+^{\partial^{out}_{3\ell} A})$. Then 
\be
\dot K(t)= \tau_t^{A^c}(D), \qtx{where} D:=i[H^{A^c},P_+^{\partial^{out}_{3\ell} A}],
\ee
and (for $t\ge 0$; the case $t\le 0$ can be treated in the same way),
\be\label{FTC}
K(t)-K(0)=\tau_t^{A^c}(P_+^{\partial^{out}_{3\ell} A})-P_+^{\partial^{out}_{3\ell} A}=  \int_0^t \tau_s^{A^c}(D)\ ds.
 \ee

Since   $P_+^{\partial_{3\ell}^{out}A}=    P_+^{\partial_{3\ell-1}^{out}A}P_+^{\partial_{3\ell}^{out}A}$, we have
\be
D=P_+^{\partial^{out}_{3\ell-1} A}\wtilde D P_+^{\partial^{out}_{3\ell-1} A}, \sqtx{where} \supp \wtilde D =\partial [ A]^{\La}_{3\ell}\sqtx{and}\wtilde D=P_-^{\partial [ A]_{3\ell}}\wtilde D P_-^{\partial [ A]_{3\ell}}.
\ee

Let us set   $B=\partial^{\La,out}_{\ell} A  \subset A^c$, so  $\partial^{\La,out}_{3\ell-1} A= [ B ]^{A^c}_{2\ell -1}$,
$\partial^\La[ A]^{\La}_{3\ell}=   \partial^{A^c} [ B ]^{A^c}_{2\ell}$. (Note one endpoint of $B$ is an endpoint of $A^c$.  We are ignoring that $A^c$ consists of possibly two intervals, we do the procedure separately on each one.)  We  can write
\be\label{Rii}
D= P_+^{ [ B ]^{A^c}_{2\ell -1}} P_-^{ \partial^{A^c} [ B ]^{A^c}_{2\ell}}\wtilde D P_-^{ \partial^{A^c} [ B ]^{A^c}_{2\ell}} P_+^{ [ B ]^{A^c}_{2\ell -1}}, \qtx{with}  \supp \wtilde D ={ \partial^{A^c} [ B ]^{A^c}_{2\ell}}.
\ee

 Since $D$ is an observable on $A^c$ of the form given in \eq{eq:Mspca}, and by the induction hypothesis 
the lemma  is true for $q- \frac 12$, it follows from Part (ii) of the theorem that there exists an observable 
$D_t=D(t,q-\frac 12,\ell)$ on $A^c$, supported on  
\beq
[B]^{A^c}_{(\beta_{q-\frac 12} +2 )\ell+1}=\partial^{\La,out}_{(\beta_{q-\frac 12} +3 )\ell+1} A,
\eeq
such that
\be\label{eq:obsO2a56}
\E\norm{P_{\check I_{\le q-1}}^{A^c}\pa{\tau_t^{A^c}(D)-D_t}P_{\check I_{\le q-1}}^{A^c}}
& \le \E\norm{P_{I_{\le q-\frac 12}}^{A^c}\pa{\tau_t^{A^c}(D)-D_t}P_{ I_{\le q-\frac 12}}^{A^c}}  \\
&\le C_{q-\frac 12}  \langle t\rangle^{p_{q-\frac 12}}\, \abs{A^c}^{ \xi_{q-\frac 12}}\e^{-\theta_{q-\frac 12}\ell},
\ee
and
\be\label{eq:outco'}
D_t=P_+^{B} D_t P_+^{B}=P_+^{\partial^{\La,out}_{\ell} A} D_t P_+^{\partial^{\La,out}_{\ell} A}.
\ee

It follows, using \eq{FTC}, that the observable on $A^c$ given by  (we now bring back the  $R$ index)
\beq\label{Yt}
Y\up{R}_t= P_+^{\partial^{\La, out,R}_{3\ell} A} +  \int_0^t \tau_s^{A^c_R}(D)\ ds,
\eeq
 is supported on $[\partial^{\La,out,R}_{\ell} A ]^{A^c_R}_{(\beta_{q-\frac 12} +3 )\ell+1}$,  and
 satisfies 
\be\label{eq:obsO2abbbb}
 \E\norm{
 P_{\check I_{\le q-1}}^{A^c_R}\pa{\tau_t^{A^c_R}(P_+^{\partial^{out,R}_{3\ell} A})-Y_t}
P_{\check I_{\le q-1}}^{A^c_R}     } 
\le C^\pr_{q-\frac 12}  \langle t\rangle^{{p_{q-\frac 12}}+1}\, \abs{A^c_R}^{ \xi_{q-\frac 12}}\e^{-\theta_{q-\frac 12}\ell},
\ee
and
\be\label{eq:outco444}
Y\up{R}_t=P_+^{\partial^{\La,out,R}_{\ell} A} Y\up{R}_t P_+^{\partial^{\La,out,R}_{\ell} A}.
\ee
We construct  $Y\up{L}_t$ in a similar way  for the interval $A^c_L$,
and define 
\beq\label{MtPA22}
M_t=M(t,q,\ell)= \wtilde M_t  Y\up{L}_t  Y\up{R}_t.
\eeq  
  It follows that
\be
\supp M_t= A\cup\partial^{\La,out}_{(\beta_{q-\frac 12} +3 )\ell+1} A =[A]^\La_{(\beta_{q-\frac 12} +3 )\ell+1}\subset [A]^\La_{(\beta_{q-\frac 12} +4 )\ell},
\ee
and
\be\label{MtPA}
M_t=  P_+^{\partial^{\La}_{\ell} A}M_t P_+^{\partial^{\La}_{\ell} A}.
\ee

We have, using    \eq{eq:incoq} and \eq{eq:obsO2abbbb}, 
\be\label{PsiMt}
&\E\norm{{ \Psi_q( H^{A,A^c}) \pa{ \tau_t^{A,A^c}\pa{ M   } - M_t}\Psi_q( H^{A,A^c}) }}  \le \E\norm{{P_{\check I_{\le q}}^{A,A^c} \pa{ \tau_t^{A,A^c}\pa{ M   } - M_t}P_{\check I_{\le q}}^{A,A^c} }}     \\
& \quad  \le 
 C_q \abs{\Lambda}^{ \xi_q}\e^{-2\theta_{q}\ell}  +  2C^\pr_{q-\frac 12}  \langle t\rangle^{{p_{q-\frac 12}}+1}\, \abs{A^c}^{ \xi_{q-\frac 12}}\e^{-\theta_{q-\frac 12}\ell} \\
  & \quad \le C_q \langle t\rangle^{{p_{q-\frac 12}}+1}
  \abs{\Lambda}^{\xi_q} e^{-\theta_{q-\frac 12}\ell}.
\ee

In addition, it follows from Lemma \ref{lem:detachaA} and   \eq{MtPA} that (recall $\Psi_q=\Psi_{q,0}$)
\be\label{PsiMPsiMt}
&\E\norm{\pa{ \Psi_q( H^{A,A^c})  M_t \Psi_q( H^{A,A^c}) -M_t}_{P^\La_{I_{\le q}}} }\\ &  \qquad =\E\norm{\pa{ \Psi_q( H^{A,A^c})  P_+^{\partial^{\La}_{\ell} A}M_t P_+^{\partial^{\La}_{\ell} A}\Psi_q( H^{A,A^c}) -P_+^{\partial^{\La}_{\ell} A}M_t P_+^{\partial^{\La}_{\ell} A}}_{P^\La_{I_{\le q}}} }\\
 &  \qquad  \le  C_q \scal{t}^3 \abs{\La}^{\xi_q}\e^{- \theta_q \ell}.
\ee 

Using  \eq{PPsiq},  \eq{PsiMt}, and \eq{PsiMPsiMt},   we get
\be\label{PPsiqMt345}
&\E\norm{\pa{\tau^\La_t \pa{ M}- M_t }_{P^\La_{I_{\le q}}} }\le \E\norm{\pa{\tau^\La_t \pa{ M}-  \Psi_q( H^{A,A^c})  \tau_t^{A,A^c}\pa{ M   } \Psi_q( H^{A,A^c}) }_{P^\La_{I_{\le q}}} } \\
&\quad + \E\norm{\pa{ \Psi_q( H^{A,A^c})  \tau_t^{A,A^c}\pa{ M   } \Psi_q( H^{A,A^c})- \Psi_q( H^{A,A^c})  M_t \Psi_q( H^{A,A^c}) }_{P^\La_{I_{\le q}}} }\\
&\qquad \qquad  + \E\norm{\pa{ \Psi_q( H^{A,A^c})  M_t \Psi_q( H^{A,A^c}) -M_t}_{P^\La_{I_{\le q}}} }\\
& \quad \le C_q \scal{t}^3 \abs{\La}^{\xi_q}\e^{-3 \theta_q \ell} + C_q \langle t\rangle^{{p_{q-\frac 12}}+1}
  \abs{\Lambda}^{\xi_q} e^{-\theta_{q-\frac 12}\ell}  +    C_q \scal{t}^3\abs{\La}^{\xi_q}\e^{- \theta_q \ell}\\
& \quad \le   C_q\langle t\rangle^{\max \set{p_{q-\frac 12}+1, 3}}   \abs{\Lambda}^{\xi_q} e^{-\theta_{q}\ell}.
\ee

We can now prove Part (i). Let   $T$ be an observable  supported on an interval $\mathcal X\subset \La$ with $\norm{T}\le 1$ satisfying \eq{P-T}.  If   $\dist \pa{\cX, \Z\setminus \La} \le\beta_q\ell$, there is nothing to prove since $[\cX]^\La_{\beta_q \ell}=\La$,  just take $T_t=\tau^\La_t(T)$. So assume $\dist \pa{\cX, \Z\setminus \La} > \beta_q\ell$.  As $\beta_q -9(\hat q +1)=\beta_{q-\frac 12}  +4\ge 4 $, we can use 
Lemma~\ref{lem:Mmod}, and  let $\what T=\sum_{j=1}^{\hat q +1} T_j$ where $T_j$ is given in \eq{defTj},
so $\norm{T_j}\le 1$, and observe that in view of   \eq{T-Tj666}
it suffices to prove Part (i) for each $T_j$, $j=1,2,\ldots ,\hat q +1$. 

Let $j=1,2,\ldots ,\hat q +1$ and set $ A_j= [\cX]_{9j\ell}$. It follows from  \eq{defTj} that
\be \label{Tjagain}
T_j=   P_+^{\partial^{out}_{3\ell} A_j} \otimes \wtilde T_j= P_+^{\partial_{3\ell} A_j}  T_j  P_+^{\partial_{3\ell} A_j}, \;
 \supp  \wtilde T_j  = A_j, \;\wtilde T_j =   P_+^{\partial^{in}_{3\ell} A_j}  \wtilde T_j  P_+^{\partial^{in}_{3\ell} A_j}.
\ee
 Moreover,  
 it follows from \eq{P-T} that
\be\label{Tjagain2}
\wtilde T_j  =  P_-^{A_j} \wtilde  T_jP_-^{A_j} \qtx{for} j=1,2,\ldots, \hat q +1.
\ee

Thus $T_j$ is an observable satisfying \eq{TjM} and 
\eq{TjM2} with $M=T_j$ and $A=A_j$.   Thus there exists an observable $(T_j)_t$ with support  
\be
[A_j]^\La_{(\beta_{q-\frac 12} +4 )\ell}=[\cX]^\La_{(\beta_{q-\frac 12} +4 +9j)\ell},
\ee
satisfying \eq{PPsiqMt345} for $T_j$.

Defining 
\beq\label{defTt}
T_t= \sum_{j=1}^{ { \hat q+1}}(T_j)_t,
\eeq 
an observable on $\La$  with
\be \label{defTt33}
\supp T_t= [\cX]^\La_{(\beta_{q-\frac 12} +4 +9(\hat q +1))\ell}= [\cX]^\La_{\beta_q\ell},
\mqtx{where} \beta_q=(\beta_{q-\frac 12} +9\hat q + 13 ),
\ee
we have, using also  \eq{T-Tj666}, 
\be\label{PPsiqMt99999}
&\E\norm{\pa{\tau^\La_t \pa{ T}- T_t }_{P^\La_{I_{\le q}}} }\le
  C_q\langle t\rangle^{p_q} \abs{\Lambda}^{\xi_q}\e^{- \theta_q \ell},\mqtx{with}p_q=\max \set{p_{q-\frac 12}+1, 3}.
\ee

Since $p_0=0$,
 we have $p_{\frac12}=3$, so $p_q= p_{q-\frac 12}+1> 3$  for $q\ge 1$.  It follows that  $p_q= 3 + 2(q-\frac 12)=  2q + 2$ for  $q\ge 1$.

 Using \eqref{betaq}, we have 
\be
\beta_{q-\frac 12} + 9q +13 \le \beta_q \le \beta_{q-\tfrac 12} + 9(q+\tfrac 12) +13 =\beta_{q-\tfrac 12} + 9q  + \tfrac {35} 2.
\ee
 Letting 
$\gamma_s= \beta_{\frac s 2}$ for $s\in \N^0$, we have $\gamma_0=0$  and $\gamma_{s-1} + \frac 92 s +  13\le \gamma_s\le \gamma_{s-1} + \frac 92 s + \frac {35} 2$ for $s\in \N$. It follows that for $s\in \N$ we have 
\be
\tfrac 92  \tfrac {s(s+1)} 2 +  13s\le \sum_{r=1}^s \pa{ \tfrac 92 r + 13 }\le 
\gamma_s \le \sum_{r=1}^s \pa{ \tfrac 92 r + \tfrac {35} 2}=  \tfrac 92  \tfrac {s(s+1)} 2 +  \tfrac {35} 2   s,
\ee
so  $  \tfrac 94 s^2 +\tfrac {61} 4 s \le    \gamma_s \le \tfrac 94 s^2 +\tfrac {79} 4 s $.
It follows that for  $ q\in  \tfrac 12\N$ we have  $ 9 q^2 + \frac {61}2 q \le\beta_q\le 9 q^2 + \frac {79}2 q$.

We now turn to Part (ii). If $T$ is is an observable on $\La$ with  $\norm{T}\le 1$ of the form given in \eq{eq:Mspca}, $T$ satisfies the hypothesis of Part (i) with $\cX=\supp T=  [ \cY ]^{\La}_{2\ell +1}$.
Thus the proof of Part (i) applies to this observable $T$, but we modify it as follows.  We observe that since
$T= P_+^{ [ \cY ]^{\La}_{\ell }} TP_+^{ [ \cY ]^{\La}_{\ell }}$  in view of \eq{eq:Mspca}, when we apply Lemma~\ref{lem:Mmod} the observable $T_j$ also satisfies  $T_j=P_+^{ [ \cY ]^{\La}_{\ell }}T_jP_+^{ [ \cY ]^{\La}_{\ell }}$ in view of \eq{defTj}.  Moreover, the corresponding observable $\wtilde T_j$ given in \eq{TjM} also satisfies $\wtilde T_j= P_+^{ [ \cY ]^{\La}_{\ell }}\wtilde T_j P_+^{ [ \cY ]^{\La}_{\ell }} $.  We define $(\wtilde T_j)_t$ as in \eq{wtildeMt}, and also define $   (\what T_j)_t=   P_+^{\cY}(\wtilde T_j)_t P_+^{\cY}$.  Since $I-P_+^{\mathcal Y}=P_-^{\cY}$
(note $A_j=[ [ \cY ]^{\La}_{2\ell +1}]^\La_{9j\ell}= [ \cY ]^{\La}_{(9j+2)\ell +1}$), we have
\be
&\E\norm{(\wtilde T_j)_t-(\what T_j)_t}\le
\E\norm{P_{\check I_{\le q}}^{A_j}\tau_t^{A_j}(\wtilde T_j)P_{\check I_{\le q}}^{A_j}-P_+^{\mathcal Y} P_{\check I_{\le q}}^{A_j}\tau_t^{A_j}(\wtilde T_j)P_{\check I_{\le q}}^{A_j}P_+^\mathcal Y}\le C_q \abs{\La}^{\xi_q}e^{-\theta_{q}\ell}
\ee
by \eqref{eq:eigencor5}.  As a result, we can replace $ (\wtilde T_j)_t$ by  $(\what T_j)_t$ for $j=1,2,\ldots, \hat q +1$  in the definition of $ (T_j)_t$ in \eq{MtPA22}, so now we have $ (T_j)_t= P_+^{\cY}(T_j)_t  P_+^{\cY}$.
As a consequence, $T_t$ defined as in \eq{defTt}  satisfies  \eq{eq:Mspca'}, and   Part (ii) is proven.  
\end{proof}

\begin{proof}[Proof of Theorem~\ref{thm:localmodell}]  Let   $q\in \frac 12\N^0$,  and assume  Condition $\cL_{q+\frac 12}$ is satisfied.  Let  $T$  be an  observable  supported on an interval $\mathcal X\subset \La$ with $\norm{T}\le 1$. As in the proof of Part (i) of Lemma \ref{thm:localmodell2}, we may assume $\dist \pa{\cX, \Z\setminus \La} >( \beta_{q+\frac12}+13 )\ell$, and   use Lemma~\ref{lem:Mmod}. Let 
$\what T=\sum_{j=1}^{\hat q +1} T_j = T_1 +\check T, \qtx{where}\check T=\sum_{j=2}^{\hat q +1} T_j $
with  $T_j$  given in \eq{defTj},
so $\norm{T_j}\le 1$, and, using \eq{T-Tj666}, 
it suffices to prove the theorem  for $T_1$ and $\check T$.

Each $T_j$,  $j=2,3,\ldots ,\hat q +1$, $T$ can  be treated  as in the proof of Part (i) of Lemma~\ref{thm:localmodell2},  see \eq{Tjagain}--\eq{PPsiqMt99999}.  Setting 
$\check T_t= \sum_{j=2}^{ { \hat q+1}}(T_j)_t$,  we get \eq{defTt33} and \eq{PPsiqMt99999} with  $\check T_t$ substituted for $T_t$. Note  that $\supp \check T=  [\cX]^\La_{\pa{9(\hat q +1)+3}\ell}$.

 Thus it only remains to prove the theorem for 
 $T_1={T} P_+^{\partial_{3\ell}[\mathcal X]_{9\ell}} =P_+^{\partial_{3\ell}[\mathcal X]_{9\ell}} {T} P_+^{\partial_{3\ell}[\mathcal X]_{9\ell}} $.   Note that $M=T_1$ satisfies \eq{TjM}--\eq{TjM46} with $A=  [\mathcal X]_{9\ell}$,  but   \eq{TjM2} may not hold.
We proceed as in the proof for $M$ in \eq{TjM} but without \eq{TjM2}, so  \eq{PAAcA} is replaced by
 \be\label{PAAcA33}
&P_{\check I_{\le q}}^{A,A^c}\tau_t^{A}(\wtilde M)\tau_t^{A^c}(P_+^{\partial^{out}_{3\ell} A})P_{\check I_{\le q}}^{A,A^c}  =P_{\check I_{\le q}}^{A,A^c}P_{\check I_{\le q}}^{A}\tau_t^{A}(\wtilde M)P_{\check I_{\le q}}^{A}\,P_{\check I_{\le q}}^{A^c}\tau_t^{A^c}(P_+^{\partial^{out}_{3\ell} A})P_{\check I_{\le q}}^{A^c}P_{\check I_{\le q}}^{A,A^c}.
 \ee

We continue the proof as before up to \eq{Rii}. $D$, defined in  \eq{Rii}, is an observable on $A^c_R$ of the form given in \eq{eq:Mspca}.  We now use that  Condition $\cL_{q+\frac 12}$ is satisfied, and hence we can apply 
Part (ii) of Lemma~\ref{thm:localmodell2} to the observable $D$  in the energy interval $I_{\le q+\frac 12}$.
 It follows  that there exists an observable 
$D_t=D(t,q+\frac 1 2,\ell)$ on $A^c_R$, supported on    
 \beq
[B]^{A^c_R}_{(\beta_{q+\frac 1 2} +2 )\ell+1}=\partial^{\La,out,R}_{(\beta_{q+\frac 1 2} +3 )\ell+1} A,
\eeq
 such that
\be\label{eq:obsO2a77}
\E\norm{\pa{\tau_t^{A^c_R}(D)-D_t}_{P^{A^c}_{\check I_{\le q}}}}&\le \E\norm{\pa{\tau_t^{A^c_R}(D)-D_t}_{P^{A^c_R}_{I_{\le q+\frac 12}}}} \\
&\le C_{q+\frac 12}  \langle t\rangle^{p_{q+\frac 12}}\, \abs{A^c}^{ \xi_{q+\frac 12}}\e^{-\theta_{q+\frac 12}\ell},
\ee
where we used  $\check I_{\le q}\subset I_{\le q+\frac 12}$, 
and
\be\label{eq:outco}
D_t=P_+^{B} D_t P_+^{B}=P_+^{\partial^{\La,out,R}_{\ell} A} D_t P_+^{\partial^{\La,out,R}_{\ell} A}.
\ee

Letting $Y_t\up{R}$ be as in \eq{Yt}, we have $\supp Y_t=  [\partial^{\La,out,R}_{\ell} A ]^{A^c_R}_{(\beta_{q+ \frac 12} +3 )\ell+1}$, 
\be\label{eq:obsO2abbbb4}
&\E\norm{P_{\check I_{\le q}}^{A^c_R}\pa{\tau_t^{A^c_R}(P_+^{\partial^{out,R}_{3\ell} A})-Y_t\up{R}}P_{\check I_{\le q}}^{A^c_R}} \le C^\pr_{q+\frac 12}  \langle t\rangle^{{p_{q+\frac 12}}+1}\, \abs{A^c}^{ \xi_{q+\frac 12}}\e^{-\theta_{q+\frac 12}\ell},
\ee
and we have \eq{eq:outco444}.

 We construct  $Y\up{L}_t$ in a similar way  for the interval $A^c_L$,
and define 
\beq\label{MtPA2255}
M_t=M(t,q,\ell)=\wtilde M_t  Y\up{L}_t  Y\up{R}_t,
\eeq  
 as in \eq{MtPA22}, so  
\be
\supp M_t= A\cup\partial^{\La,out}_{(\beta_{q+\frac 12} +3 )\ell+1} A =[A]^\La_{(\beta_{q+\frac 12} +3 )\ell+1}\subset [A]^\La_{(\beta_{q+\frac 12} +4 )\ell}= [\cX]^\La_{(\beta_{q+\frac 12} +13)\ell},
\ee
and we have \eq{MtPA}.

We have, using    \eq{eq:incoq} and \eq{eq:obsO2abbbb4}, 
\be\label{PsiMt4}
&\E\norm{{ \Psi_q( H^{A,A^c}) \pa{ \tau_t^{A,A^c}\pa{ M   } - M_t}\Psi_q( H^{A,A^c}) }}  \le \E\norm{{P_{\check I_{\le q}}^{A,A^c} \pa{ \tau_t^{A,A^c}\pa{ M   } - M_t}P_{\check I_{\le q}}^{A,A^c} }}     \\
& \quad  \le 
 C_q \abs{\Lambda}^{ \xi_q}\e^{-2\theta_{q}\ell}  + 2C^\pr_{q+\frac 12}  \langle t\rangle^{{p_{q+\frac 12}}+1}\, \abs{A^c}^{ \xi_{q+\frac 12}}\e^{-\theta_{q+\frac 12}\ell}\\
  & \quad \le C_{q+\frac 12} \langle t\rangle^{{p_{q+\frac 12}}+1}
  \abs{\Lambda}^{\max \set{\xi_q,\xi_{q+\frac 12} }}
 e^{-\min\set{2\theta_{q},\theta_{q+\frac 12}}\ell}\\
 & \quad \le C_{q+\frac 12} \langle t\rangle^{{p_{q+\frac 12}}+1}
  \abs{\Lambda}^{\xi_{q+\frac 12} }
 e^{-\theta_{q+\frac 12}\ell}.
\ee

In addition,  \eq{PsiMPsiMt} holds. 
Using  \eq{PPsiq},  \eq{PsiMt4}, and \eq{PsiMPsiMt},   we get
\be\label{PPsiqMt345'}
&\E\norm{\pa{\tau^\La_t \pa{ M}- M_t }_{P^\La_{I_{\le q}}} }\le \E\norm{\pa{\tau^\La_t \pa{ M}-  \Psi_q( H^{A,A^c})  \tau_t^{A,A^c}\pa{ M   } \Psi_q( H^{A,A^c}) }_{P^\La_{I_{\le q}}} } \\
&\quad + \E\norm{\pa{ \Psi_q( H^{A,A^c})  \tau_t^{A,A^c}\pa{ M   } \Psi_q( H^{A,A^c})- \Psi_q( H^{A,A^c})  M_t \Psi_q( H^{A,A^c}) }_{P^\La_{I_{\le q}}} }\\
&\qquad \qquad  + \E\norm{\pa{ \Psi_q( H^{A,A^c})  M_t \Psi_q( H^{A,A^c}) -M_t}_{P^\La_{I_{\le q}}} }\\
& \quad \le C_q \scal{t}^3 \abs{\La}^{\xi_q}\e^{-3 \theta_q \ell} +  C_{q+\frac 12} \langle t\rangle^{{p_{q+\frac 12}}+1}
  \abs{\Lambda}^{\xi_{q+\frac 12} }
 e^{-\theta_{q+\frac 12}\ell} +    C_q \scal{t}^3\abs{\La}^{\xi_q}\e^{- \theta_q \ell}\\
& \quad \le  C_{q+\frac 12}\langle t\rangle^{\max \set{p_{q+\frac 12}+1, 3}} \abs{\Lambda}^{\max \set{\xi_q,\xi_{q+\frac 12} }}\e^{- \theta_{q+\frac 12} \ell}\\
& =  C_{q+\frac 12}\langle t\rangle^{{p_{q+\frac 12}+1}} \abs{\Lambda}^{{\xi_{q+\frac 12} }}\e^{- \theta_{q+\frac 12} \ell}=  C_{q+\frac 12}\langle t\rangle^{2q+4} \abs{\Lambda}^{{\xi_{q+\frac 12} }}\e^{- \theta_{q+\frac 12} \ell}.
\ee

Setting $T_t= \what T_t +M_t$,  then $T_t$ is  supported in  $ [\cX]^\La_{(\beta_{q+\frac 12} +13)\ell}$ and, using  now  \eq{T-Tj666},  $T_t$  satisfies 
\eq{eq:locality2mgq}.
\end{proof}

\section{Proof of slow propagation of information, matrix elements version}\label{secpropmatrix}

\begin{proof}[Proof of Corollary \ref{cor:localmodell}]
 Let   $q\in \frac 1 2 \N$,  and assume  Condition $\cL_{q+1}$ is satisfied.   Set  
 \be\label{alpha}
  \alpha = \max \set{ 2 \tfrac {\xi_{q+\frac 12} } { \theta_{q+\frac 12}}, \frac{4\cl{q}} {c_\mu } +1},
  \ee
and consider a finite interval  $\La \subset \Z$ with $\abs{\Lambda}$ sufficiently large so 
\be\label{rcond}
 r    + 2(4r+2)r = 8r^2 + 5 r< \abs{\La}, \qtx{where} r=\cl{\alpha \ln\abs{\Lambda}}.   
 \ee
Let  $\ell\in \N$,  $t\in\R$, and  
$M_1, M_2 \subset \La$ with   $\abs{M_1}=\abs{M_2}=\what N\in [1,\abs{\La}]\cap \N$.

Suppose $\ell\ge r$.  In this case we
 pick $T_t$ as in Theorem~\ref{thm:localmodell}, so it follows from \eq{eq:locality2mgq} and \eq{alpha} that
\be\label{PiM1M2q1}
\E\norm{\pi_{M_1}\pa{\tau^\La_t(T)-{T}_t}_{P^\La_{I_{\le q}}}\pi_{M_2}}\le \E\norm{\pa{\tau^\La_t(T)-{T}_t}_{P^\La_{I_{\le q}}}}\le    C_{q}\langle t\rangle^{2q+4} e^{- \frac 1 2\theta_{q+\frac 12}\, \ell}.
\ee

Thus we only need to consider the case   $ \ell < r$. Suppose first that $ \what N\ge r $. In this case we use a large deviation argument.
On the complement of the event ${\mathcal B_{\cl{q}}^{\what N}}$ we have $\chi_\La^{\what N}P^\La_{I_{\le q}}=0$  (see \cite[Eqs. (3.52 and (3.55)]{EK22}), and hence, taking $T_t=T$ we have
\be\label{LD111}
\E\norm{\pi_{M_1}\pa{\tau^\La_t(T)-{T}_t}_{P^\La_{I_{\le q}}}\pi_{M_2}}\le   2 \P \pa{\mathcal B_{\cl{q}}^{\what N}}
\le C_q\abs{\Lambda}^{2\cl{q}} \e^{- c_\mu\what  N} \le  C_q  \e^{- \frac 1 2 c_\mu r}\le C_q \e^{- \frac 1 2 c_\mu\ell},
\ee
where  in the second step we used the large deviation estimate \cite[Eq. (3.53)]{EK22} and  \eq{alpha}.

 It remains to consider  the case  $\what N < r$ and $\ell < r$.  It follows that $\abs{M_1\cup M_2} \le 2(r -1)$.
 We assume $\abs{\cX} \le  \ln \abs{\La}< r $  (note $\alpha>1$ by \eq{alpha}).
 It follows from \eq{rcond} that 
 there exists $j\in[0,2r]\cap \N^0$ such that 
 \be
 \pa{[\cX]^\La_{(2j+2)r}\setminus [\cX]^\La_{2jr}}\cap \pa{M_1 \cup M_2}=\emptyset,
 \ee  
and  set  $\cX_r:=[\cX]^\La_{(2j+1)\ell}$,  and observe that
 \be
 \abs{\cX_r} \le  \abs{\cX}+ 2(4r+1)r \le r +  2(4r+1)r= 8r^2 +3 r \le 11 r^2.
 \ee
 
Since $T$ is supported by $\cX\subset \cX_r$,  we  use Theorem \ref{thm:localmodell}  with $\La=\cX_r$ and $q+\frac 12$ instead of $q$, concluding that
there exists  an observable $T_\ell=T(t,q+\frac 12 ,\ell,\cX_r)$ supported by 
  $ [\cX]^{\cX_r}_{(13+\beta_{q+1} )\ell}\subset  [\cX]^{\La}_{(13+\beta_{q+1} )\ell}$,
such that 
\be\label{tautTt4}
&\E\norm{\pa{\tau^{\cX_r}_t(T)-{T}_t}_{P^{\cX_r}_{\check I_{\le q}}}} \le  \E\norm{\pa{\tau^{\cX_r}_t(T)-{T}_t}_{P^{\cX_r}_{I_{\le q+\frac 12}}}} 
\\  &  \quad \le C_{q+\frac 12}\langle t\rangle^{2q+5}  { \abs{\cX_r}}^{{\xi_{q+1} }}\e^{- {\theta_{q+1}}\ell}  \le C_{q+\frac 12}\langle t\rangle^{2q+5} \pa{\ln (11 r^2)}^{{\xi_{q+1} }}\e^{- {\theta_{q+1}}\ell}\\
& \quad \le C^\pr_{q}\langle t\rangle^{2q+5} \pa{\ln r}^{{\xi_{q+1} }}\e^{- {\theta_{q+1}}\ell}
\le C^{\prr}_{q}\langle t\rangle^{2q+5} \pa{\ln \abs{\La}}^{{\xi_{q+1} }}\e^{- {\theta_{q+1}}\ell}.
\ee

 Since   $\pi_{M_i}=\pi_{M_i}P_+^{\partial^\La_{r}\cX_r}= P_+^{\partial^\La_{r}\cX_r}\pi_{M_i}$, $i=1,2$, 
 we use Lemma \ref{lemHS}  and an argument similar to the proof  of Lemma \ref{lem:detachaA}   to deduce that
\be\label{Mtaupsi1}
&\E\norm{\pi_{M_1}{\pa{\tau^\La_t(T) -\Psi_q( H^{\cX_r, \cX_r^c})\tau^{\cX_r, \cX_r^c}_t(T)\Psi_q( H^{\cX_r, \cX_r^c})}}_{P^\La_{I_{\le q}}}\pi_{M_2}}\\
&  \quad  =\E\norm{\pi_{M_1}P_+^{\partial^\La_{r}\cX_r}{\pa{\tau^\La_t(T) -\Psi_q( H^{\cX_r, \cX_r^c})\tau^{\cX_r, \cX_r^c}_t(T)\Psi_q( H^{\cX_r, \cX_r^c})}}_{P^\La_{I_{\le q}}}P_+^{\partial^\La_{r}\cX_r}\pi_{M_2}}\\
&  \quad  \le \E\norm{P_+^{\partial^\La_{r}\cX_r}{\pa{\tau^\La_t(T) -\Psi_q( H^{\cX_r, \cX_r^c})\tau^{\cX_r, \cX_r^c}_t(T)\Psi_q( H^{\cX_r, \cX_r^c})}}_{P^\La_{I_{\le q}}}P_+^{\partial^\La_{r}\cX_r}}\\      &  \quad \le  C_q \scal{t}^3 \abs{\La}^{\xi_q}\e^{- \theta_q r}\le   C_q \scal{t}^3 \e^{- \frac 12\theta_q r},
\ee
where we used \eq{alpha} and $ \frac {\xi_{q} } { \theta_{q}} \le  \frac {\xi_{q+ \frac 12} } { \theta_{q+\frac zz1 2}}$.

We note that  $\tau^{\cX_r, \cX_r^c}_t(T)= \tau^{\cX_r}_t(T)$ on $\cH_\La$ since 
 $T$ is supported by $\cX\subset \cX_r$.  Thus
 \be\label{Mtaupsi2}
  &\E\norm{\pi_{M_1}{\pa{\Psi_q( H^{\cX_r, \cX_r^c})\pa{\tau^{\cX_r, \cX_r^c}_t(T)-T_t}\Psi_q( H^{\cX_r, \cX_r^c})}}_{P^\La_{I_{\le q}}}\pi_{M_2}}\\
  & \quad \le  \E\norm{{\pa{\Psi_q( H^{\cX_r, \cX_r^c})\pa{\tau^{\cX_r, \cX_r^c}_t(T)-T_t}\Psi_q( H^{\cX_r \cX_r^c})}}_{P^{\cX_r, \cX_r^c}_{\check  I_{\le q}}}}\\
  & \quad \le  \E\norm{{\pa{\tau^{\cX_r,}_t(T)-T_t}}_{P^{\cX_r}_{\check I_{\le q}}}} \le  C^{\prr}_{q}\langle t\rangle^{2q+5} \pa{\ln \abs{\La}}^{{\xi_{q+1} }}\e^{- {\theta_{q+1}}\ell},
  \ee
 where we used $P^{\cX_r, \cX_r^c}_{\check  I_{\le q} }=P^{\cX_r}_{\check  I_{\le q} }P^{\cX_r, \cX_r^c}_{\check  I_{\le q} }$ and \eq{tautTt4}
 
 It follows from \eq{Mtaupsi1} and \eq{Mtaupsi2} that
 \be\label{PiM1M2q2}
&\E\norm{\pi_{M_1}\pa{\tau^\La_t(T)-{T}_t}_{P^\La_{I_{\le q}}}\pi_{M_2}} \\
& \quad \le   \E\norm{\pi_{M_1}{\pa{\tau^\La_t(T) -\Psi_q( H^{\cX_r, \cX_r^c})\tau^{\cX_r, \cX_r^c}_t(T)\Psi_q( H^{\cX_r, \cX_r^c})}}_{P^\La_{I_{\le q}}}\pi_{M_2}}\\
& \qquad\quad  +  \E\norm{\pi_{M_1}{\pa{\Psi_q( H^{\cX_r, \cX_r^c})\pa{\tau^{\cX_r, \cX_r^c}_t(T)-T_t}\Psi_q( H^{\cX_r, \cX_r^c})}}_{P^\La_{I_{\le q}}}\pi_{M_2}}\\
& \quad \le   C_q \scal{t}^3 \e^{- \frac 12\theta_q r} + 
C^{\prr}_{q}\langle t\rangle^{2q+5} \pa{\ln \abs{\La}}^{{\xi_{q+1} }}\e^{- {\theta_{q+1}}\ell} \le C^{\prr\pr}_{q}\langle t\rangle^{2q+5} \pa{\ln \abs{\La}}^{{\xi_{q+1} }}\e^{-\min\set{ \frac 12 \theta_q,\theta_{q+1}}\ell}.
\ee

We conclude from \eq{PiM1M2q1},  \eq{LD111}, and \eq{PiM1M2q2} that for all $\ell \in \N$ there exists an observable ${T}_t={T}(t,q,\ell, \La, M_1\cup M_2)$, supported in  $ [\cX]^\La_{(13+\beta_{q+1} )\ell}$,  such that we have
\be\label{eq:locality2mgqcor9933}
\E\norm{\pi_{M_1}\pa{\tau^\La_t(T)-{T}_t}_{P^\La_{I_{\le q}}}\pi_{M_2}} \le C_{q}\langle t\rangle^{2q+5} \pa{\ln \abs{\La}}^{\xi_{q+1} }\e^{-\frac 12 \min\set{ \theta_q,\theta_{q+\frac 1 2}, c_\mu}\ell}.
\ee 
\end{proof}

\section*{Declarations}

\subsection*{\qquad Data availability}
We do not analyse or generate any datasets, because our work proceeds within a theoretical and mathematical approach. One can obtain the relevant materials from the references below.

\subsection *{\qquad  Funding and/or Conflicts of interests/Competing interests} \

Alexander Elgart was  supported in part by the NSF under grant DMS-1907435 and the Simons Fellowship in Mathematics Grant 522404.

The authors have no relevant financial or non-financial interests to disclose.

The authors have no competing interests to declare that are relevant to the content of this article.

\printbibliography

\end{document}